\documentclass{article}

\usepackage{microtype}
\usepackage{graphicx}
\usepackage{subfigure}
\usepackage{booktabs} 

\usepackage{hyperref}

\usepackage[accepted]{icml2023}

\usepackage{amsmath}
\usepackage{amssymb}
\usepackage{mathtools}
\usepackage{amsthm}

\usepackage[capitalize,noabbrev]{cleveref}

\theoremstyle{plain}
\newtheorem{theorem}{Theorem}[section]

\newtheorem{lemma}[theorem]{Lemma}

\theoremstyle{definition}
\newtheorem{definition}[theorem]{Definition}

\theoremstyle{remark}

\usepackage[textsize=tiny]{todonotes}

\usepackage{multirow}
\usepackage{booktabs}
\usepackage{math_commands}
\usepackage{graphicx}
\usepackage{setspace}
\usepackage{colortbl}
\usepackage{xcolor}
\usepackage{setspace}
\usepackage{thm-restate}

\icmltitlerunning{End-to-End Full-Atom Antibody Design}

\begin{document}

\twocolumn[
\icmltitle{End-to-End Full-Atom Antibody Design}



\icmlsetsymbol{equal}{*}

\begin{icmlauthorlist}
\icmlauthor{Xiangzhe Kong}{thucs,air}
\icmlauthor{Wenbing Huang}{ruc,lab}
\icmlauthor{Yang Liu}{thucs,air}
\end{icmlauthorlist}

\icmlaffiliation{thucs}{Dept. of Comp. Sci. \& Tech., Institute for AI, BNRist Center, Tsinghua University}
\icmlaffiliation{air}{Institute for AI Industry Research (AIR), Tsinghua University}
\icmlaffiliation{ruc}{Gaoling School of Artificial Intelligence, Renmin University of China}
\icmlaffiliation{lab}{Beijing Key Laboratory of Big Data Management and Analysis Methods, Beijing, China}

\icmlcorrespondingauthor{Wenbing Huang}{hwenbing@126.com}
\icmlcorrespondingauthor{Yang Liu}{liuyang2011@tsinghua.edu.cn}

\icmlkeywords{End-to-End, Full-Atom, Antibody Design, E(3)-Equivariance}

\vskip 0.3in
]



\printAffiliationsAndNotice{}  

\begin{abstract}
Antibody design is an essential yet challenging task in various domains like therapeutics and biology. There are two major defects in current learning-based methods: 1) tackling only a certain subtask of the whole antibody design pipeline, making them suboptimal or resource-intensive.  2) omitting either the framework regions or side chains, thus incapable of capturing the full-atom geometry. To address these pitfalls, we propose \textbf{dy}namic \textbf{M}ulti-channel \textbf{E}quivariant gr\textbf{A}ph \textbf{N}etwork (dyMEAN), an end-to-end full-atom model for E$(3)$-equivariant antibody design given the epitope and the incomplete sequence of the antibody. Specifically, we first explore \emph{structural initialization} as a knowledgeable guess of the antibody structure and then propose \emph{shadow paratope} to bridge the epitope-antibody connections. Both 1D sequences and 3D structures are updated via an \emph{adaptive multi-channel equivariant encoder} that is able to process protein residues of variable sizes when considering full atoms. Finally, the updated antibody is docked to the epitope via the alignment of the shadow paratope. Experiments on epitope-binding CDR-H3 design, complex structure prediction,  and affinity optimization demonstrate the superiority of our end-to-end framework and full-atom modeling. 
\end{abstract}

\vspace{-0.2in}
\section{Introduction}
\label{sec:intro}

Antibodies are a family of Y-shaped proteins in immune systems that binds to pathogens, commonly called antigens, with specificity~\citep{raybould2019five}. Antibody design for target epitopes on the antigen exhibits tremendous potential and necessity in therapeutic and biological research~\citep{tiller2015advances, almagro2018progress, yuan2020highly}. Nevertheless, the task is challenging because the \textit{complementarity determining regions (CDRs)}, where the binding mainly occurs, are highly variant, and the underlying regularity of antigen-antibody interactions is arduous to unveil. The past decade has seen the application of traditional energy-based optimization~\citep{li2014optmaven, adolf2018rosettaantibodydesign}, learning-based language models on the 1D sequence~\citep{liu2020antibody, saka2021antibody}, as well as recent deep generative methods to co-design the CDR sequences and 3D structures simultaneously, exhibiting appealing superiority over conventional sequence-based approaches.~\citep{jin2021iterative, luo2022antigen, kong2022conditional}.

\begin{figure}[t]
    \centering
    \includegraphics[width=1.0\columnwidth]{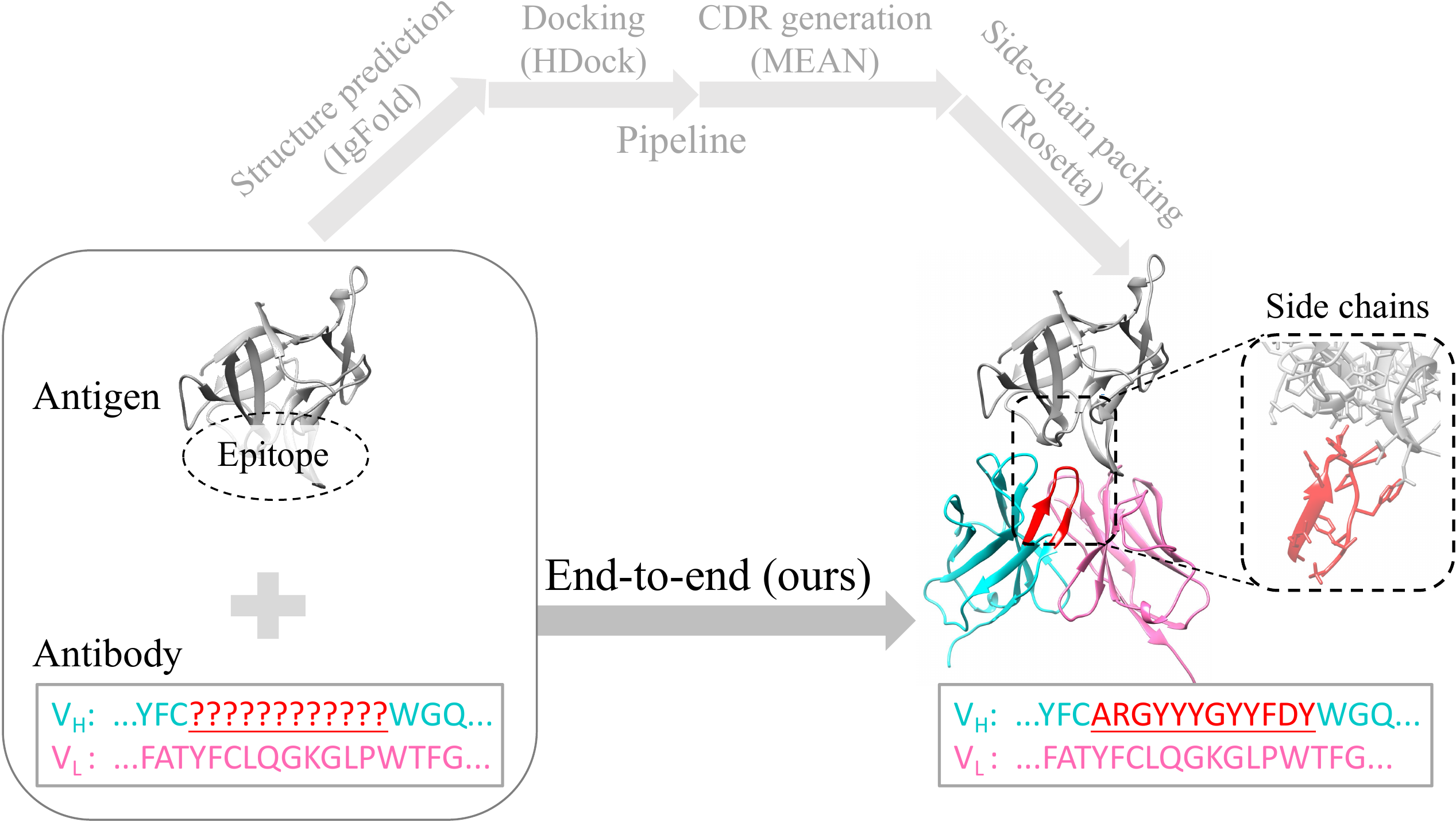}
    \vskip -0.2in
    \caption{Our end-to-end full-atom antibody design. By contrast, current computational methods resort to the multi-stage solution: e.g., IgFold~\citep{ruffolo2022fast} for structure prediction, HDock~\citep{yan2020hdock} for docking on the target epitope, MEAN~\citep{kong2022conditional} for binding CDR generation, and Rosetta~\citep{alford2017rosetta} for side-chain packing.}
    \label{fig:task}
    \vskip -0.2in
\end{figure}

Despite the impressive progress, current computational models are still incapable of fulfilling the real need for antibody design. In most practical cases, we only know the 3D structure of antigen with the target epitope and the 1D incomplete sequence (without CDRs) of antibody. To address this ill-posed task, a potential computational pipeline includes: structure prediction~\citep{ruffolo2022fast}, antigen-antibody docking~\citep{yan2020hdock},  binding CDR generation~\citep{jin2022antibody, luo2022antigen, kong2022conditional}, and side-chain packing~\citep{alford2017rosetta}, as illustrated in Figure~\ref{fig:task}. Existing works can solve each local problem separately, but lacks the mastery of the global picture, making them suboptimal. Conducting wet-lab experiments, such as obtaining the antigen-antibody complex structure from cryo-electron microscopy, somehow bypasses this suboptimality, yet is much more costly and laborious~\citep{carter2006potent}. Therefore, the deficiency of both computational pipelines and experimental methods poses an urgent need for a computational end-to-end solution.


Furthermore, the full-atom geometry is critical for depicting the interactions within the antigen-antibody complex~\citep{foote1992antibody, jones1996principles}. Current works usually model the backbone atoms only~\citep{jin2021iterative, kong2022conditional}, or simply consider the orientation of side chains~\citep{luo2022antigen}. Although \citet{jin2022antibody} makes an initial attempt to incorporate all side-chain atoms into a hierarchical graph, it suffers from efficiency problems and is obliged to omit all other components of the antibody except CDR-H3 (see Appendix~\ref{app:complexity}), leading to incomplete context modeling and thus inaccurate design. Full-atom geometry of the entire antibody has a much larger scale and demands a computationally more efficient and effective model.

To address the above two issues, we propose \textbf{dy}namic \textbf{M}ulti-channel \textbf{E}quivaraint gr\textbf{A}ph \textbf{N}etwork (dyMEAN) as an end-to-end and full-atom solution. Compared to previous works~\citep{luo2022antigen, kong2022conditional}, we directly tackle the end-to-end problem where only the epitope and the incomplete 1D sequence are known in advance (Figure~\ref{fig:task}), in contrast to previous multi-stage solutions. We explore knowledge-guided \emph{structural initialization} based on conserved residues and propose \textit{shadow paratope} to capture antigen-antibody interaction that is invariant to their initial orientations and positions. The 1D sequence and the 3D structure are updated iteratively via an \emph{adaptive multi-channel message passing}, which favorably tolerates the variance in the number of channels (\emph{i.e.}, atoms) in different residues, when considering full-atom geometry. We finally achieve epitope-antibody docking through the alignment of the shadow paratope. The network also conforms to E(3)-equivariance, which is a critical property exhibited in 3D biology~\citep{kong2022conditional}. Experiments on epitope-binding CDR-H3 design, complex structure prediction, and affinity optimization demonstrate the superiority of our end-to-end framework and full-atom modeling.

\vspace{-0.1in}
\section{Related Work}
\label{sec:related_work}
\paragraph{Antibody Design} Conventional computational methods commonly optimize sophisticated energy functions designed by domain experts~\citep{li2014optmaven, adolf2018rosettaantibodydesign}, or train language models on the 1D sequences~\citep{liu2020antibody, saka2021antibody, akbar2022silico}. Energy-based methods suffer from the insufficient expressive power of the statistical energy functions~\citep{mackerell2002charmm, leaver2011rosetta3}, and language models are suboptimal due to the lack of structural modeling.
More recently, the community has witnessed the emergence of sequence-structure co-design methods and their superiority over previous methods~\citep{jin2021iterative, jin2022antibody, luo2022antigen, kong2022conditional}. However, they are limited to certain stages of pipeline-based antibody design. For example, \citet{jin2021iterative} generates the CDRs on a single chain, and \citet{luo2022antigen,kong2022conditional} fill in the CDRs given a docked complex, demanding hard-to-obtain prerequisites. \citet{jin2022antibody} attempts to generate and dock CDR-H3 simultaneously on the local binding interface. Nevertheless, it suffers from the inefficiency of the distance-based initialization, the hierarchical encoding, and the autoregressive refinement (see Appendix~\ref{app:complexity}), preventing it from scaling to the entire antibody. Distinct from the above works, we directly generate the complete complex given the epitope and the incomplete sequence in an end-to-end and full-atom manner.
\vspace{-0.1in}
\paragraph{Protein Docking} Generally, protein docking predicts the docked complex of two proteins given their unbound structures~\citep{kozakov2017cluspro, yan2020hdock, ganea2021independent}. While they require the structure of both proteins in advance, our work simultaneously generates the structure of the antibody and docks it to the antigen. Another difference lies in the prior knowledge of the binding regions on the antigen and the antibody (\emph{i.e.}, the epitope and the paratope). Only certain epitopes on the antigen constitute meaningful targets in therapeutics~\citep{yuan2020highly}, and the paratope mostly comes from CDRs, especially CDR-H3~\citep{kuroda2012computer}. Therefore, antibody docking mainly focuses on the local binding interface, while many protein docking methods (e.g., EquiDock, \citealp{ganea2021independent}) assume no prior knowledge of the epitope and the paratope, making them suboptimal in this situation.
\vspace{-0.1in}
\paragraph{Equivariant Graph Neural Networks}
Equivariant graph neural networks are designed with the desired inductive bias that the results should not rely on the view of observation, namely E(3)-equivariance. With increasing availability of 3D data, abundant equivariant neural networks have emerged~\citep{thomas2018tensor, gasteiger2020directional, fuchs2020se, satorras2021n}. Our work is closely related to the multi-channel equivariant graph networks proposed by \citet{kong2022conditional}, where each residue node has multiple coordinates (i.e., channels) referring to different atoms. We propose a more powerful version of multi-channel equivariant message passing, which is adaptive to the variable number of channels in full-atom modeling.

\vspace{-0.1in}
\section{Notations and Definitions}
\label{sec:def}

\begin{figure}[htbp]
    \centering
    \vskip -0.1in
    \includegraphics[width=1.0\columnwidth]{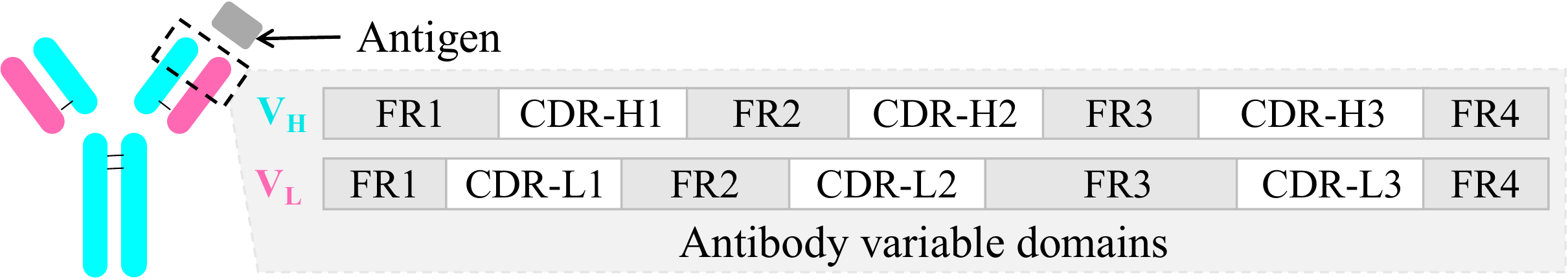}
    \vskip -0.15in
    \caption{Variable domains in the heavy/light chain ($V_H$ / $V_L$).}
    \label{fig:preliminary}
    \vskip -0.1in
\end{figure}

\begin{figure*}[!t]
    \centering
    \vskip -0.1in
    \includegraphics[width=\textwidth]{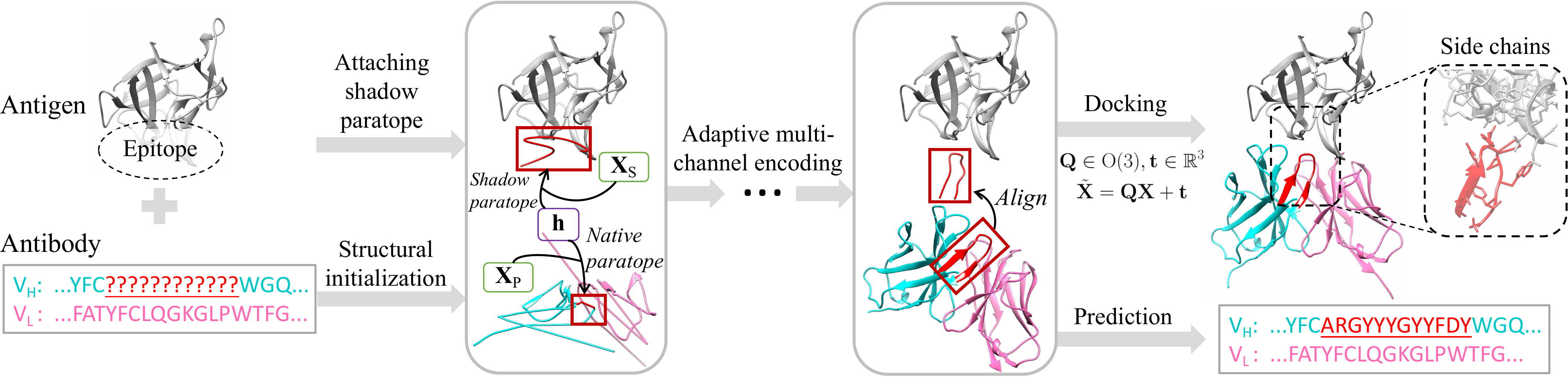}
    \vskip -0.15in
    \caption{Overall architecture. \textbf{Structural initialization} (\textsection~\ref{sec:init}): obtaining the initial hidden vector $\vh_i$ and coordinate matrix $\mX_i$ for the antibody. \textbf{Attaching shadow paratope} (\textsection~\ref{sec:shadow}): attaching a clone of the paratope around the epitope, where $\vh_i$ is shared, but the coordinates are private for the shadow paratope and the native one, \emph{i.e.}, $\mX_S$ v.s. $\mX_P$. \textbf{Adaptive multi-channel encoding} (\textsection~\ref{sec:dyenc}): updating $\vh_i$ and $\mX_i$ by multi-channel message passing, where the full-atom geometry is characterized. \textbf{Docking} (\textsection~\ref{sec:e2e}): Aligning the native paratope to the shadow paratope. \textbf{Prediction} (\textsection~\ref{sec:e2e}): outputting the amino acid type of each residue in the paratope.} 
    \label{fig:model}
    \vskip -0.15in
\end{figure*}

A protein comprises one or more long chains of amino acid residues. An antibody is a Y-shaped symmetric protein with two identical sets of chains, as illustrated in Figure~\ref{fig:preliminary}. Each set contains a heavy chain and a light chain, either of which consists of several constant domains and a variable domain. As their names suggest, the constant domains keep unchanged across different antibodies; while the variable domain varies to enable different binding specificity for different antigens, making it the main focus of antibody design. We denote the variable domains of the heavy chain and the light chain by $V_H$ and $V_L$, respectively. The variable domain is further divided into alternating arrangements of four \textit{framework regions (FRs)} and three \textit{complementarity determining regions (CDRs)}. The binding regions of an antigen and an antibody are called an \textit{epitope} and a \textit{paratope}, separately. In this paper, the paratope refers to CDR-H3 in the heavy chain following~\citet{jin2022antibody}, since it is highly variable and dominates binding~\citep{maccallum1996antibody}.

We describe the epitope of the antigen and the variable domains of the antibody as the graphs $\gG_E(\gV_E,\gE_E)$ and $\gG_A(\gV_A,\gE_A)$, where $\gV_E$ and $\gV_A$ refer to the vertices (\emph{i.e.}, the residues), $\gE_E$ and $\gE_A$ are edges.  Each residue $v_i$ is represented by its amino acid type $s_i$ and a multi-channel 3D coordinate matrix $\mX_i \in \R^{3 \times c_i}$, where $c_i$ denotes the channel size, \emph{i.e.}, the number of atoms in $v_i$. Notably, previous studies~\citep{kong2022conditional} only consider backbone atoms for each residue, and the coordinate dimension is constant: $c_i=4$. This paper models the full-atom geometry by further involving side chains, hence $c_i$ is distinct for different residues.  
The edges are constructed by finding the $k$-Nearest Neighbors (kNN) of each residue, using the minimum pair-wise distance between all atoms in $v_i$ and $v_j$: 
\vskip -0.2 in
\begin{align}
\label{eq:dist-X}
d(v_i, v_j) = \min_{1\leq p \leq c_i, 1 \leq q \leq c_j}||\mX_{i}(:,p) - \mX_{j}(:,q)||_2,
\end{align}
\vskip -0.1 in
where $\mX_{i}(:,p)$ returns the $p$-th atom in $\mX_i$ and $\mX_{i}(:,q)$ is similarly defined.  
Inspired by \citet{kong2022conditional}, we insert three global nodes into the heavy chain, the light chain, and the epitope, respectively, connecting to all nodes in their own chains. Besides, the global nodes of the heavy chain and the light chain are linked to each other. 
\vspace{-0.1in}
\paragraph{Task Definition} The residue vertices of the paratope are denoted as $\gV_P$, clearly $\gV_P\subseteq\gV_A$. Given an epitope $\gG_E(\gV_E, \gE_E)$ and an incomplete antibody sequence $\{ s_i | i\in\gV_A, i \notin \gV_P\}$, we aim to design a model that simultaneously generates the 1D sequence of the paratope as well as the entire 3D structure of the antibody $(\gV_A, \gE_A)$ binding to the epitope, namely $\{ s_i | i \in \gV_P\}$ and $\{\mX_i | v_i \in \gV_A\}$.

\section{Our Method: dyMEAN}
\label{sec:method}

The overall workflow of our dyMEAN is presented in Figure~\ref{fig:model}. During the calculation in dyMEAN, each vertex in epitope graph $\gG_E$, antibody graph $\gG_A$, and paratope subgraph $\gG_P$ ($\gG_P\subseteq\gG_A $) is associated with an invariant vector $\vh_i\in\R^d$ and an equivariant coordinate matrix $\mX_i\in\R^{3\times c_i}$. In form, the \textbf{overview of dyMEAN} is given by:
\begin{align}
    \gG_A &= \text{SI}(\{ s_i\}_{i\in\gV_A, i \notin \gV_P}), i\in\gV_A, \\
    \gG_S &= \text{SP}(\gG_E, \gG_P) \\
    \vh_i, \mX_i &= \text{AME}(\gG_E,\gG_S,\gG_A), i\in\gV_E\cup\gV_S\cup\gV_A,\\
    \label{eq:predict}
    \vp_i &= \text{Predict}(\vh_i), i\in\gV_P,\\
    \label{eq:dock}
    \tilde{\mX}_i & = \text{Dock}(\gG_A,\gG_S), i\in\gV_A,
\end{align}
where, $\text{SI}$ (\emph{a.k.a.} structural initialization) first initializes the coordinates $\mX_i^{(0)}$ and the hidden states $\vh_i^{(0)}$ for the antibody graph  $\gG_A$, based on the the incomplete antibody sequence; with the initialized paratope $\gG_P$, $\text{SP}$ (\emph{a.k.a.} shadow paratope) attaches a shadow paratope $\gG_S$, which shares the hidden states with the native one, to the epitope $\gG_E$, creating a joint graph $\gG_E\cup\gG_S$. SP is crucial in bridging the epitope and the antibody for docking. Then, $\text{AME}$ (\emph{a.k.a.} adaptive multi-channel encoder) iteratively updates $\mX_i$ and $\vh_i$ for all vertices by message passing; finally, Eq.~\ref{eq:predict} predicts the distribution of amino acid types $\vp_i$ for each paratope residue, and Eq.~\ref{eq:dock} docks the antibody $\gG_A$ towards the shadow paratope $\gG_S$, leading to the binding complex structure $\tilde{\mX}_i$.   

We will introduce the functions of $\text{SI}$, $\text{SP}$ and 
$\text{AME}$ in \textsection~\ref{sec:init}, \textsection~\ref{sec:shadow} and \textsection~\ref{sec:dyenc}, respectively. The details of Eq.~\ref{eq:predict}, Eq.~\ref{eq:dock}, and training losses are provided in \textsection~\ref{sec:e2e}. An elegant property of dyMEAN is that its predicted paratope sequence is invariant and the binding structure is equivariant, with respect to the E(3) transformations (rotations/reflections/translations), making it well generalizable to different poses of the target epitope. We will mathematically reveal this point in \textsection~\ref{sec:e2e}.


\subsection{Structural Initialization with Conserved Residues}
\label{sec:init}

The input antibody sequence $\{ s_i\}_{i\in\gV_A, i \notin \gV_P}$ involves neither paratope information nor the 3D geometry. Given such fragmentary information, this subsection investigates on how to attain desirable initialization for both $\vh_i^{(0)}$ and $\mX_i^{(0)}$.

\textbf{Initializing $\vh_i^{(0)}$}~
We derive the initial embedding of each node via its amino acid type $s_i$ and position number $r_i$ in an numbering system ,\emph{e.g.}, IMGT~\citep{lefranc2003imgt}:
$\vh_i^{(0)} = \vf(s_i, r_i) = \vf_{s_i} + \vf_{r_i}$, where $\vf_{s_i}$ and $\vf_{r_i}$ define the learnable amino acid embedding and position embedding, respectively. For the unknown paratope residue, we represent $s_i$ by a special type $\mathrm{[MASK]}$.

\textbf{Initializing $\mX_i^{(0)}$}~
We have the domain knowledge that the FRs of the antibody are well conserved~\citep{klein2013somatic} in  spatial variation. It inspires us to first detect the well-conserved residues in FRs and then apply them to sketch the positions of other residues. While it is challenging to directly locate the conserved residues by comparing the residue-wise coordinates, we resort to the comparison of 1D sequences, which innately reflects the 3D spatial similarity~\citep{jumper2021highly}. To do so, we first align the antibody sequences in the dataset via a certain antibody numbering system (\emph{e.g.}, IMGT). Then we consider a residue as well-conserved if its type is consistent among above $95\%$ of the antibodies the analysis of different thresholds is provided in Appendix~\ref{app:conserve_th}). Next, we align all the antibodies by the backbone (i.e. $N, C_\alpha, C, O$) coordinates of these well-conserved residues via Kabsch algorithm~\citep{kabsch1976solution}, and calculate the average backbone coordinates of these residues, leading to the backbone template $\{\mZ_{r_i} \in \R^{3 \times 4}| r_i \in \sW\}$, where $\sW$ collects the position numbers of the detected well-conserved residues. In our experiments, we identify 16 such residues in the heavy chain and 18 in the light chain. The backbone coordinates $\mZ_i$ of other residues in the same chain are valued in this way: (1) for the ones between two  nearest conserved residues in position number, we linearly interpolate their positions with unified spacing; (2) for those located at both ends of the chain, we conduct outwards linear-interpolation from the nearest conserved residue with the same interval used in the nearest pair of the residues computed in (1).
More details are provided in Appendix~\ref{app:init}.
$\mZ_i$ is then extended to $\mX_i^{(0)}$ by filling $\alpha$-carbon's coordinate in the side chains . We emphasize the significance of this knowledgeable initialization, which provides vague but essential guess of the antibody structure.

The coordinates are further normalized to conform to the standard Gaussian distribution $\gN(0, \mI)$, by conducting 3D mean translation and 1D variance normalization (all dimensions of all antibodies share the same normalization factor to ensure consistent scale). After obtaining $\mX_i^{(0)}$, we construct the kNN edges for $\gG_A$ via the distance defined in Eq.\ref{eq:dist-X}.  


\subsection{E(3)-Invariant Attachment of Shadow Paratope}
\label{sec:shadow}
We attach a clone of the paratope around the epitope, which is called \textit{shadow paratope}. It serves two crucial purposes in our end-to-end framework: (1) Transmitting E(3)-invariant information between the epitope and the antibody, by sharing the hidden states $\vh_i$ and the same topology with the native paratope; (2) acting as the key points that will be used for the docking between the antibody and epitope, which will be detailed in~\textsection~\ref{sec:e2e}. 
One promising property of our shadow paratope attachment is that its 3D coordinates and final docked structure are independent of the initial position of the antibody, since it only exchanges the invariant information (\emph{i.e.}, $\vh_i$ not $\mX_i$) with the native paratope.

The shadow paratope subgraph is $\gG_S = (\gV_S, \gE_S)$. Here, $\gE_S$ contains two parts: internal edges copied from the connections between residues in the native paratope, and external edges linked to the epitope. For $v_i \in \gV_E, v_j \in \gV_S$, the external edges are constructed based on the kNN distance: 
\vskip -0.2 in
\begin{align}
    \label{eq:pdist}
    \hat{d}(v_i, v_j) = \phi_e(\vh_i, \vh_j) + \phi_e(\vh_j, \vh_i),
\end{align}
\vskip -0.1 in
where $\phi_e$ is a Multi-Layer Perceptron (MLP). The hidden vector $\vh_i$ of $\gG_S$ is duplicated from the native paratope, and the coordinates $\mX_i$ are initialized around the center of the epitope according to standard Gaussian $\gN(0, \mI)$\footnote{The coordinates of the epitope have been normalized to $\gN(0, \mI)$ beforehand in a similar way to the antibody in~\textsection~\ref{sec:init}.}. $\gG_S$ is merged into the epitope graph $\gG_E$, creating $\gG_E\cup\gG_S$.

\subsection{Adaptive Multi-Channel Equivariant Encoder}
\label{sec:dyenc}
AME is able to handle $\mX_i$ of different channel size, in order to consider full-atom geometry by involving side chains besides backbone atoms. This is why we call AME adaptive. The $l$-th layer updates the hidden vector $\vh_i$ and the coordinate matrix $\mX_i$ as follows:
\begin{align}
    \label{eq:mp1}
    \vm_{ij} &= \phi_m(\vh_i^{(l)}, \vh_j^{(l)}, \frac{T_R(\mX_i^{(l)}, \mX_j^{(l)})}{||T_R(\mX_i^{(l)}, \mX_j^{(l)})||_F + \epsilon}), \\
    \label{eq:mp2}
    \mX_{ij} &= T_S(\mX_{i}^{(l)} - \frac{1}{c_j} \sum_{k = 1}^{c_j} \mX_j^{(l)}(:, k), \phi_x(\vm_{ij})), \\
    \label{eq:mp3}
    \vh_i^{(l+1)} &= \phi_h(\vh_i^{(l)}, \sum\nolimits_{j \in \gN(i)} \vm_{ij}), \\
    \label{eq:mp4}
    \mX_i^{(l+1)} &= \mX_i^{(l)} + \frac{1}{|\gN(i)|} \sum\nolimits_{j \in \gN(i)} \mX_{ij}
\end{align}
where, $\phi_m, \phi_x, \phi_h$ are MLPs, $\gN(i)$ denotes $i$'s  neighbors, $\vm_{ij}$ and $\mX_{ij}$ are non-geometric and geometric messages, respectively; the geometric relation extractor $T_R$ and geometric message scaler $T_S$ are for processing message between two distinct-shape matrices $\mX_i\in \R^{3\times c_i}$ and $\mX_j\in \R^{3\times c_j}$; the output of $T_R$ is normalized with Frobenius norm following~\citet{huang2022equivariant}, plus a constant $\epsilon=1$ for numerical stability. Below are the details of  $T_R$ and $T_S$, and how information is exchanged between the epitope and the antibody.

\paragraph{Geometric Relation Extractor $T_R$} Given $\mX_i \in \R^{3\times c_i}$ and $\mX_j \in \R^{3\times c_j}$, we first compute the channel-wise distance between each pair of the channels in $\mX_i$ and $\mX_j$: $\mD_{ij}(p, q) = ||\mX_i(:, p) - \mX_j(:, q)||_2$.
Then, we employ two learnable weights $\vw_i \in \R^{c_i \times 1}$ and $\vw_j \in \R^{c_j \times 1}$ to characterize the channel-wise correlation in $\mD_{ij}$, and two learnable attribute matrices $\mA_i\in\R^{c_i\times d}$ and  $\mA_j\in\R^{c_j\times d}$ to extract useful patterns across each channel and output dimension (further details in Appendix~\ref{app:channel_attr}). The final output $\mR_{ij}\in\R^{d\times d}$ is given by:
\begin{align}
    \label{eq:tr}
    \mR_{ij} = \mA_i^\top (\vw_i \vw_j^\top \odot \mD_{ij})\mA_j.
\end{align}
Clearly, $\mR_{ij}$ keeps the same shape regardless of the change in $c_i$ or $c_j$, namely static-dimensional inputs for $\phi_m$ and $\phi_h$.

\paragraph{Geometric Message Scaler $T_S$} 
The main purpose of $T_S$ is to generate geometric messages by scaling the input coordinates $\mX \in \R^{3 \times c}$  with the non-geometric message $\vs=\phi_x(\vm_{ij})\in\R^{C}$ where $C$ is the upper bound of the channel size. In detail, $T_S(\mX, \vs)$ is calculated by:
\begin{align}
    \mX' = \mX \cdot \mathrm{diag}(\vs'),
\end{align}
where $\vs'\in\R^{c}$ is the average pooling of $\vs$ with the window size $C-c+1$ and stride $1$, $\mathrm{diag}(\cdot)$ returns the matrix with the input vector as the diagonal elements, and thus the output $\mX'$ shares the same shape with $\mX$.

\paragraph{Information Exchanging between $\gG_E$ and $\gG_A$}
Although the epitope graph $\gG_E$ and the antibody graph $\gG_A$ are disconnected, their information is exchanged via the hidden states of the shadow paratope $\gG_S$. In particular, we first conduct 1-layer AME on $\gG_A$ and copy the hidden vectors $\vh_i$ from the native paratope $\gG_P$ to the shadow paratope $\gG_S$. Then, we carry out 1-layer AME on  $\gG_E\cup\gG_S$ and copy the hidden vectors reversely from $\gG_S$ to $\gG_P$. The above two stages are alternated until $L$ layers. We additionally run 1-layer message passing on $\gG_A$ to broadcast the updated information across the entire antibody. 

Nicely, $T_R$ is E$(3)$-invariant, $T_S$ is O$(3)$-equivariant, and the information exchanging between $\gG_E$ and $\gG_A$ is E$(3)$-invariant, therefore for the final outputs of AME, $\vh_i$ is E$(3)$-invariant and $\mX_i$ is independently E$(3)$-equivariant~\citep{ganea2021independent} \emph{w.r.t.} $\gG_E\cup\gG_S$ and $\gG_A$. Such property will permit E(3)-invariance of dyMEAN stated in Theorem~\ref{the:e3}.

\subsection{Prediction, Docking and Training Losses}
\label{sec:e2e}

 \paragraph{Prediction} 
With the output by AME, we leverage the progressive full-shot decoding strategy from~\citet{kong2022conditional} to generate the 1D sequence and the 3D structure over $T$ iterations. To be specific, each iteration updates the hidden states and the coordinates for all vertices:
\begin{align}
\label{eq:EMA-t}
\{\vh_i^{(t)},\mX_i^{(t)}\} = \text{AME}(\{\vh_i^{(t-1)},\mX_i^{(t-1)}\}).
\end{align}
We predict the amino acid type of the paratope with $\vh_i^{(t)}$, :
\begin{align}
    \vp^{(t)}_i = \mathrm{Softmax}(\phi_p(\vh^{(t)}_i)),  i \in \gV_P,
\end{align}
where $\phi_p$ is an MLP. The hidden states are refreshed as:
\begin{align}
\label{eq:memory}
\vh^{(t)}_i = \left\{
\resizebox{0.7\columnwidth}{!}{$\displaystyle{
    \begin{array}{ll}
     \vf(s_i, r_i) + \phi_d(\vh_i^{t}), & i\notin\gV_P, \\
     \sum_{j=1}^{n_a} p^{(t)}_{i,j}\vf(s_j, r_i) + \phi_{d}(\vh^{(t)}_i), & i\in\gV_P,
    \end{array}
    }$}
\right.
\end{align}
where the embedding $\vf(s_i, r_i)=\vf_{s_i}+\vf_{r_i}$, $\phi_d$ is an MLP, $n_a$ is the number of amino acid types, and $p^{(t)}_{i,j}$ returns the $j$-th element of $\vp_i^{(t)}$. The second line aims at performing soft smoothing of the embeddings with the predicted probability $\vp_i^{(t)}$. Compared with MEAN~\citep{kong2022conditional}, the memory term $\phi_{d}(\vh^{(t)}_i)$ is extra added for better information reservation, which will be ablated in our experiments. 

The new $\vh^{(t)}_i$, along with $\mX_i^{(t)}$ from Eq.~\ref{eq:EMA-t} will be used as the input for the next iteration. After each iteration, we recreate the edges $\gE_E, \gE_A, \gE_S$ by calculating the distance in Eq.~\ref{eq:dist-X} and Eq.~\ref{eq:pdist} based on the current values $\mX_i^{(t)}$ and $\vh^{(t)}_i$. 


\paragraph{Docking} After the final iteration, we align the pose of the native paratope with the shadow paratope via Kabsch algorithm~\citep{kabsch1976solution}. The docked coordinates $\{\tilde{\mX}_i | v_i \in \gV_A\}$ are given by:
\begin{align}
    \label{eq:dock1}
    &\resizebox{0.85\columnwidth}{!}{$\mQ, \vt = \mathrm{Kabsch}(\{\mX_i^{(T)} | i \in \gV_P \}, \{\mX^{(T)}_i| i \in \gV_S\})$}, \\
    \label{eq:dock2}
    &\tilde{\mX}_i = \mQ \mX^{(T)}_i + \vt, v_i \in \gV_A,
\end{align}
where $\mQ \in \text{O}(3), \vt\in\R^3$, and $\text{O}(3)$ is the orthogonal group. 

\paragraph{Loss Function} The loss function sums up the three parts: sequence loss $\gL_\text{seq}$, structure loss $\gL_\text{struct}$ and docking loss $\gL_\text{dock}$. The cross-entropy loss $\ell_{\text{ce}}$ is utilized to guide the sequence prediction at each iteration:
\begin{align}
    \gL_{\text{seq}} = \frac{1}{T|\gV_P|}\sum\nolimits_{t=1}^T \sum\nolimits_{i \in \gV_P} \ell_{\text{ce}}(\vp^{(t)}_i, \vp^\star_i).
\end{align}
\vskip -0.1in

For structure supervision, we exert Huber loss~\citep{huber1992robust} on the coordinates of the final iteration. As suggested by ~\citet{kong2022conditional}, Huber loss maintains numerical stability for noisy data (further details in Appendix~\ref{app:huber}):
\begin{align}
    \gL_{\text{coord}}  &= \frac{1}{|\gV_A|} \sum\nolimits_{v_i \in \gV_A} \ell_{\text{huber}}(\mX^{(T)}_i, \mX^\star_i),
\end{align}
where $\mX^\star_i$ denotes the ground-truth coordinates aligned to $\mX^{(T)}_i$ by Kabsch algorithm. Since our method generates the structure of all atoms, we further supervise bond lengths to capture the local geometry:
\begin{align}
    \gL_{\text{bond}}   &= \frac{1}{|\gB|} \sum\nolimits_{b \in \gB} \ell_{\text{huber}}(b^{(T)}, b^\star),
\end{align}
where $\gB$ contains all chemical bonds in the antibody, $b^{(T)}$ and $b^\star$ denote the bond length derived from $\mX^{(T)}_i$ and the ground truth, respectively. The structure loss is the sum of the above two losses: $\gL_{\text{struct}} = \gL_{\text{coord}} + \gL_{\text{bond}}$.

For docking, it is sufficient to supervise the shadow paratope by the coordinate loss and the external distance loss:

\vskip -0.2in
\begin{align}
\label{eq:L-cood}
    \gL_{\text{sp}} &= \frac{1}{|\gV_S|} \sum\nolimits_{i \in \gV_S} \ell_\text{huber}({\mX}^{(T)}_i, {\mX}^\star_i), \\
    \label{eq:L-dist}
    \gL_{\text{dist}} &= \frac{1}{T|\gV_E||\gV_S|}\sum_{t=1}^T \displaystyle{\sum_{u\in\gV_E, v\in\gV_S}} \ell_\text{huber}(\hat{d}^{(t)}(u, v), d^\star(u, v)),
\end{align}
\vskip -0.2in

where $\hat{d}^{(t)}$ defined in Eq.~\ref{eq:pdist} computes the external edge distance at $t$-th iteration and $d^\star$ is the ground-truth distance. The docking loss becomes: $\gL_\text{dock} = \gL_\text{sp} + \gL_\text{dist}$.

We now demonstrate an elegant property of our dyMEAN: it is E(3)-equivariant with respect to the initial position and orientation of the epitope. 
\begin{restatable}{theorem}{equivariance}
    \label{the:e3}
    Given the initial epitope (along with the shadow paratope) $\{\vh_i,\mX_i\}_{i\in\gV_E\cup\gV_S}$ and the initialized antibody $\{\vh_i^{(0)},\mX_i^{(0)}\}_{i\in\gV_A}$, we compute the final prediction and docking by $\{\vp_i\}_{i\in\gV_P}, \{\tilde{\mX_i}\}_{i\in\gV_A}=\text{dyMEAN}\left(\{\vh_i,\mX_i\}_{i\in\gV_E\cup\gV_S}, \{\vh_i^{(0)},\mX_i^{(0)}\}_{i\in\gV_A}\right)$. We immediately have the conclusion that $\text{dyMEAN}$ is E(3)-equivariant. Namely, for any transformations $g_1, g_2\in \text{E}(3)$, we have $\{\vp_i\}_{i\in\gV_P}, \{g_1\cdot\tilde{\mX_i}\}_{i\in\gV_A}=\text{dyMEAN}\left(\{\vh_i,g_1\cdot\mX_i\}_{i\in\gV_E\cup\gV_S}, \{\vh_i^{(0)},g_2\cdot\mX_i^{(0)}\}_{i\in\gV_A}\right)$
     where $g \cdot \mX \coloneqq \mQ \mX + \vt$ for orthogonal transformation $\mQ \in \text{O}(3)$ and translation transformation $\vt \in \R^3$.
\end{restatable}
The proof is provided in Appendix~\ref{app:e3}. This theorem is crucial, as it tells that our dyMEAN is well generalizable to arbitrary orientation and position of the epitope as well as the initialized antibody, and it is thus data-efficient.

\section{Experiments}
\label{sec:exps}

We conduct experiments on the three tasks: (1) Epitope-binding CDR-H3 generation (\textsection~\ref{sec:cdrh3}); (2) Complex structure prediction (\textsection~\ref{sec:cplx}); (3) Affinity optimization (\textsection~\ref{sec:affopt}). We also try designing binders on general proteins and provide the results in Appendix~\ref{app:cath}. Following \citet{kong2022conditional}, we extract the 48 residues closest to the antibody as the epitope, which is sufficient to include all binding residues in the antigen~\citep{shan2022deep}.

Since there is no previous method for end-to-end full-atom antibody design, we implement each subtask of the whole pipeline (structure prediction$\Rightarrow$docking$\Rightarrow$CDR generation$\Rightarrow$side-chain packing) with existing competitive approaches. For antibody structure prediction, we select the official implement of \textbf{IgFold}~\citep{ruffolo2022fast} that is a specialization of 
AlphaFold~\citep{evans2022protein} for the antibody domain. For docking, we leverage \textbf{HDock}~\citep{yan2020hdock}, which is a prevailing model with knowledge-based scoring functions. For CDR generation, the following baselines are implemented: \textbf{RosettaAb}~\citep{adolf2018rosettaantibodydesign} searches for the optimal sequence and structure guided by statistical energy functions; \textbf{MEAN}~\citep{kong2022conditional} generates both 1D sequences and 3D structures with equivariant attention graph networks; \textbf{Diffab}~\citep{luo2022antigen} is a diffusion-based generative model and has considered side-chain orientations. To further involve side chains, we use Rosetta~\citep{alford2017rosetta} to cope with side-chain packing, which is also a built-in step of RosettaAb. Besides, we implement \textbf{HERN}~\citep{jin2022antibody} that needs no external structure prediction, docking, and side-chain packing but is unaware of framework region modeling and inefficient in autoregressive generation of all atoms. 
Further implementation details are deferred to Appendix~\ref{app:exp}.


\subsection{Epitope-binding CDR-H3 Generation}
\label{sec:cdrh3}
The experiments here test the central goal of end-to-end antibody design, as illustrated in Figure~\ref{fig:task}. 
As CDR-H3 is the most variant region among all CDRs and largely determines the binding specificity and affinity~\citep{raybould2019five}, we consider it as the paratope to be generated. We also provide the analysis for designing multiple CDRs as well as the entire antibody in \textsection~\ref{sec:designall}.

We use the following metrics for quantitative assessment: \textbf{Amino Acid Recovery (AAR)} is defined as the overlapping ratio of the generated sequence and the ground truth; \textbf{CAAR}~\citep{ramaraj2012antigen} computes AAR restricted to binding residues whose minimum distance from epitope residues is below 6.6 \mbox{\normalfont\AA}; \textbf{TMscore}~\citep{zhang2004scoring, xu2010significant} measures the global similarity between the generated structure and the ground truth in terms of $C_\alpha$ coordinates; \textbf{Local Distance Difference Test (lDDT)}~\citep{mariani2013lddt} contrasts the difference of the atom-wise distance matrix between the generated structure and the ground truth; \textbf{RMSD} calculates the Root Mean Square Deviation regarding the absolute coordinates of CDR-H3 without Kabsch alignment; \textbf{DockQ}~\citep{basu2016dockq} is a comprehensive score for the docking quality. Both TMscore and lDDT range from 0 to 1 and are invariant to E(3)-transformations of the antibody structure, while RMSD and DockQ focus on the docking quality and are sensitive to the relative position of the antibody to the epitope.

We train all models on the Structural Antibody Database (SAbDab,~\citealp{dunbar2014sabdab}) retrieved in November 2022, and assess them with the RAbD benchmark~\citep{adolf2018rosettaantibodydesign} composed of 60 diverse complexes selected by domain experts. We split SAbDab into the training and validation sets with a ratio of $9:1$ according to CDR-H3 clusters as suggested by \citet{jin2021iterative, kong2022conditional}. Each cluster is formed by antibodies sharing above 40\% CDR-H3 sequence identity calculated by the BLOSUM62 substitution matrix~\citep{henikoff1992amino}. The antibodies in the same clusters as the test set are dropped to maintain a convincing generalization test. We implement the clustering process with MMseqs2~\citep{steinegger2017mmseqs2} and the numbers of antibodies (clusters) in the training and the validation sets are 3,256 (1,644) and 365 (182).
\vspace{-0.25in}
\paragraph{Results} As shown in Table~\ref{tab:cdr_gen}, our dyMEAN remarkably outperforms all baselines regarding nearly all metrics, supporting its superiority in recovering 1D sequences, 3D structures, and the binding interface.
In contrast to the pipeline-based models (RosettaAb$^\ast$, DiffAb$^\ast$, and MEAN$^\ast$), dyMEAN is end-to-end and able to alleviate potential accumulated errors incurred by each stage of the antibody design process, hence leading to better performance. Compared with HERN which is unaware of frame region modeling, dyMEAN is clearly more advantageous in both 1D generation and docking, indicating that characterizing the full-context geometry in antibody design is useful and even indispensable. In addition, the TMscore and lDDT of the initialized structure via SI are meaningful but still far from satisfactory, which explains the importance of later message passing by AME in dyMEAN. As an illustrated example, Figure~\ref{fig:cdrh3} visualizes the comparison between MEAN$^\ast$ and dyMEAN. More samples are provided in Appendix~\ref{app:sample}. We further analyze the distribution of the $\chi$-angles of the generated side chains in Appendix~\ref{app:xangles}.

\begin{table}[t!]
\centering
\vskip -0.1in
\caption{Results of epitope-binding CDR-H3 design on RAbD. Methods with superscript $\ast$ adopt the pipeline: IgFold $\Rightarrow$ HDock $\Rightarrow$ CDR generation $\Rightarrow$ Rosetta side-chain packing.}
\label{tab:cdr_gen}
\setlength\tabcolsep{5pt}
\scalebox{0.70}{
\begin{tabular}{cccccccc}
\toprule
\multirow{2}{*}{Model}                             & \multicolumn{3}{c}{Generation}                       &  & \multicolumn{3}{c}{Docking}              \\ \cline{2-4} \cline{6-8} 
                                                   & AAR$\uparrow$    & TMscore$\uparrow$& lDDT$\uparrow$ &  & CAAR$\uparrow$& RMSD$\downarrow$& DockQ$\uparrow$\\ \midrule
 RosettaAb$^*$    & 32.31\%          & 0.9717          & 0.8272          &  & 14.58\% & 17.70         & 0.137          \\
DiffAb$^*$ & 35.31\%          & 0.9695          & 0.8281          &  & 22.17\% & 23.24         & 0.158          \\
MEAN$^*$   & 37.38\%          & 0.9688          & 0.8252          &  & 24.11\% & 17.30         & 0.162          \\
HERN                                               & 32.65\%          & -               & -               &  & 19.27\% & \hphantom{0}9.15   & 0.294          \\ \hline
Initialization                               & -                & 0.5072          & 0.2998          &  & -                & -             & -              \\
dyMEAN                                       & \textbf{43.65\%} & \textbf{0.9726} & \textbf{0.8454} &  & \textbf{28.11\%} & \textbf{\hphantom{0}8.11} & \textbf{0.409} \\ \bottomrule
\end{tabular}
}
\vskip -0.1in
\end{table}

\subsection{Complex Structure Prediction}
\begin{figure}[t]
    \centering
    \includegraphics[width=.9\columnwidth]{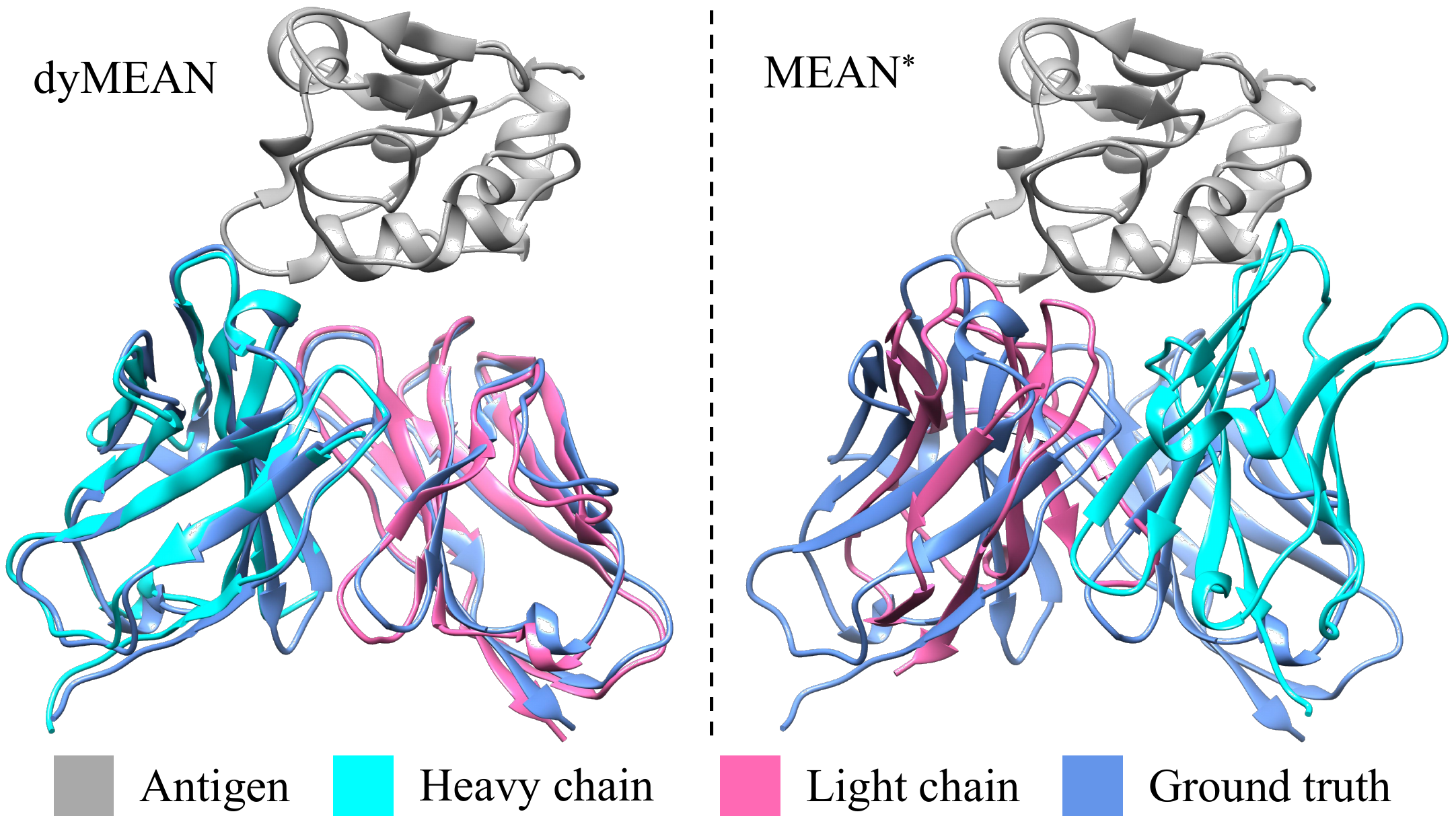}
    \vskip -0.15in
    \caption{Complexes (pdb: 1ic7) generated by our dyMEAN (DockQ$=0.971$) and MEAN$^\ast$ (DockQ$=0.046$).}
    \label{fig:cdrh3}
    \vskip -0.1in
\end{figure}
\label{sec:cplx}
This task predicts the docked complex structure given the complete antibody sequence (including CDR-H3). We report the metrics of TMscore, lDDT, RMSD, and DockQ. As there is no need for CDR generation, the pipeline-based method is reduced as: IgFold$\Rightarrow$Hdock$\Rightarrow$ Rosetta. To better depict the effectiveness of our method, we also implement the docking version of HERN in two considerate ways: (1)  taking input as the backbone structure predicted by IgFold, HERN outputs the docked backbones, followed by Rosetta for side-chain packing; (2) taking input as the ground-truth antibody structures, HERN docks CDR-H3 along with other regions towards the epitope. We train all models on SAbDab with the training-validation ratio of 9:1 and evaluate on the test set (51 antigen-antibody complexes) used in IgFold paper~\citep{ruffolo2022fast} to avoid any potential data leakage when applying IgFold during testing.
\vspace{-0.1in}
\paragraph{Results} Table~\ref{tab:cplx_pred} reads that dyMEAN surpasses all other methods in terms of both structure prediction and docking. Excitingly, though IgFold leverages embeddings from a pretrained antibody language model~\citep{ruffolo2021deciphering} and utilizes 38k additional antibody structures from AlphaFold~\citep{jumper2021highly} for training, our model still achieves better TMscore and lDDT, exhibiting its stronger capability of learning the distribution of antibody structures. 
As for the baseline GT$\Rightarrow$HERN that applies ground-truth structures for docking, our dyMEAN still yields better docking accuracy, which reveals that dyMEAN really excels at unveiling the epitope-antibody interactions with the full-context geometry. We also explore including other CDRs into the shadow paratope which is presented in Appendix~\ref{app:mulcdr_dock}.d

\begin{table}[htbp]
\centering
\vskip -0.15in
\caption{Complex Structure Prediction. Methods with superscript $\ast$ use Rosetta to generate the side chains. The values with $\dag$ are the upper-bound as they are calculated on Ground Truths (GT).}
\label{tab:cplx_pred}
\scalebox{0.8}{
\begin{tabular}{cccccc}
\toprule
\multirow{2}{*}{Model}          & \multicolumn{2}{c}{Structure}     &           & \multicolumn{2}{c}{Docking}    \\ \cline{2-3} \cline{5-6} 
                                & TMscore$\uparrow$& lDDT$\uparrow$&           & RMSD$\downarrow$& DockQ$\uparrow$ \\ \midrule
IgFold$\Rightarrow$HDock$^\ast$ & 0.9701\hphantom{0}& 0.8439\hphantom{0}          &           & 16.32         & 0.202          \\
IgFold$\Rightarrow$HERN$^\ast$  & 0.9702\hphantom{0}& 0.8441\hphantom{0}          &           & \hphantom{0}9.63          & 0.429          \\
GT$\Rightarrow$HERN             & 1.0000$^\dag$   & 1.0000$^\dag$   &           & \hphantom{0}9.65          & 0.432          \\ \hline
initialization            & 0.5054\hphantom{0}    & 0.3006\hphantom{0}          & \textbf{} & -             & -              \\
dyMEAN                    & \textbf{0.9731\hphantom{0}} & \textbf{0.8673\hphantom{0}} & \textbf{} & \textbf{\hphantom{0}9.05} & \textbf{0.452} \\ \bottomrule
\end{tabular}
}
\vskip -0.1in
\end{table}

\subsection{Affinity Optimization}
\label{sec:affopt}
Another common application is to optimize the affinity of a given antibody. As suggested by \citet{kong2022conditional}, we use the binding affinity change ($\Delta\Delta G$) as the objective, which is predicted by a GNN-based predictor~\citep{shan2022deep}. We also provide the results using FoldX~\citep{schymkowitz2005foldx} as the affinity predictor in Appendix~\ref{app:aff_foldx}. We conduct evaluation on the antibodies from SKEMPI V2.0~\citep{jankauskaite2019skempi}. We also report the number of changed residues $\Delta L$ since many practical scenarios prefer smaller $\Delta L$~\citep{ren2022proximal}. To adjust dyMEAN for this task, we additionally train an MLP over the representations of the complex graphs to fit the above-mentioned $\Delta\Delta G$ predictor. Then we conduct gradient search to locate favorable initial states of all residues, which are likely to generate a complex of higher affinity. A few more adaptions are needed, which are detailed in Appendix~\ref{app:prop_opt}. For compared baselines, we use ITA for MEAN, and the intermediate state at the $(T-t)$-th step during the denoising process for DiffAb, as suggested in their papers. All models are trained on SAbDab under the same settings as~\textsection~\ref{sec:cdrh3}. For each antibody in the test set, we generate 100 candidates and record the $\Delta\Delta G$ of the top-1 candidate, and then compute the corresponding $\Delta L$. 
\vspace{-0.2in}
\paragraph{Results} Table~\ref{tab:aff_opt} summarizes the average $\Delta\Delta G$ and $\Delta L$ over all test antibodies. It shows that dyMEAN generates antibodies with the lowest $\Delta \Delta G$ and controllable changes of $\Delta L$. 
Although DiffAb can also control $\Delta L$ by reducing $t$, its ability to affinity optimization is limited. MEAN achieves favorable $\Delta\Delta G$ but at the cost of great change in $\Delta L$. It is worth mentioning that our model still achieves desirable performance even when only 1 or 2 residues are allowed to change. Figure~\ref{fig:opt} (right) illustrates one example in this case. 

\begin{table}[t!]
\centering
\vskip -0.1in
\caption{Average $\Delta \Delta G$ (kcal/mol) and average number of changed residues ($\Delta L$). dyMEAN-$n$ denotes the restricted version allowing at most $n$ changed residues, and dyMEAN itself changes $n$ residues with $n$ is sampled from $[1,N]$ at each generation.}
\label{tab:aff_opt}
\setlength\tabcolsep{3pt}
\scalebox{0.8}{\begin{tabular}{ccc}
\toprule
Method       & $\Delta \Delta G\downarrow$& $\Delta L\downarrow$ \\ \midrule
Diffab ($t=1$)\hphantom{0}  & -0.32             & 1.19       \\
Diffab ($t=2$)\hphantom{0}  & -0.68             & 1.21       \\
Diffab ($t=4$)\hphantom{0}  & -1.00             & 1.38       \\
Diffab ($t=8$)\hphantom{0}  & -1.34             & 1.62       \\
Diffab ($t=16$)             & -1.85             & 3.54       \\
Diffab ($t=32$)             & -2.17             & 7.06       \\ \bottomrule
\end{tabular}}
\hspace{1em}
\scalebox{0.8}{\begin{tabular}{ccc}
\toprule
Method       & $\Delta \Delta G\downarrow$ & $\Delta L\downarrow$ \\ \midrule
MEAN         & -6.48             & 8.96       \\
dyMEAN-$1$     & -6.79             & 1.00       \\
dyMEAN-$2$     & -7.11             & 1.59       \\
dyMEAN-$4$     & -7.18             & 3.24       \\
dyMEAN-$8$     & -7.23             & 6.67       \\ 
dyMEAN\hphantom{-0} & \textbf{-7.31}    & 5.57       \\ \bottomrule
\end{tabular}}
\vskip -0.15in
\end{table}

\begin{figure}[t]
    \centering
    \includegraphics[width=\columnwidth]{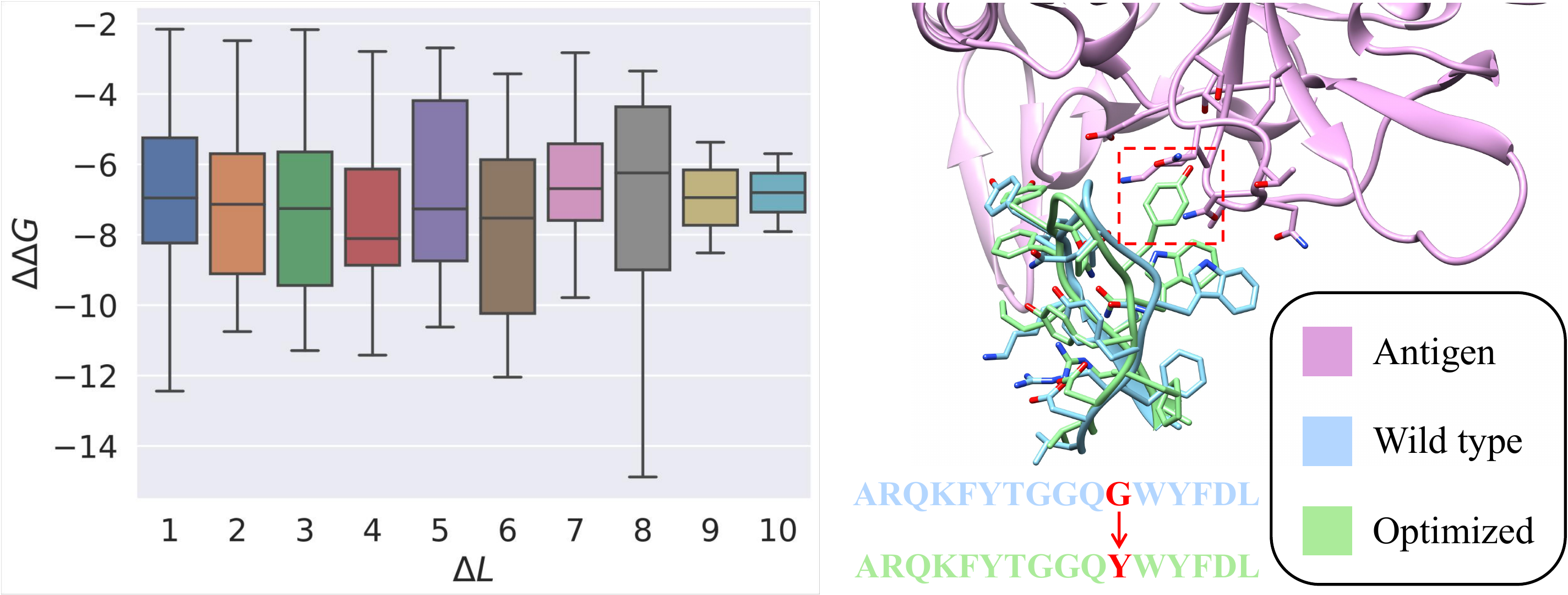}
    \vskip -0.15in
    \caption{Left: The distribution of $\Delta\Delta G$ w.r.t. the actual $\Delta L$ of the candidates after optimization. Right: The binding interface of an optimized antibody (pdb: 3se9, $\Delta\Delta G=-7.22$) with only one residue changed compared to the wild type.}
    \label{fig:opt}
    \vskip -0.2in
\end{figure}

\section{Analysis}
\label{sec:analysis}
\paragraph{Ablation Study} We ablate the necessity or value of the following components: the number of iterations in generation $T$, the full-atom geometry, the information sharing between the shadow and the native paratope, the learnable channel weights $\vw_i\vw_j^\top$ in Eq.~\ref{eq:tr}, the memory term $\phi_d$ in Eq.~\ref{eq:memory}, and the external distance prediction loss $\gL_\text{dist}$ in Eq.~\ref{eq:L-dist}. Particularly, for the ablation of the full-atom geometry, we only retain backbone atoms in dyMEAN and use Rosetta for side-chain packing afterward; for $\gL_\text{dist}$, we instead use the coordinates to compute distances between residue pairs other than predicting them with hidden states. 
Table~\ref{tab:ablation} presents the following observations: (1)
The value of $T$ mainly affects the docking performance, and $T=3$ used in dyMEAN generally yields the best performance. (2) The removal of the full-atom geometry exerts a remarkably adverse impact on the overall performance, which confirms the necessity of incorporating the side-chain conformation. (3) The information sharing is critical specifically for structure generation and docking, without which all metrics excluding CAAR drop by a large margin. (4) The learnable weights act like attentions to different channels, and will incur detriment if removed. (5) The memory mechanism contributes to the task of CDR-H3 design but seems nonessential on complex structure prediction, which is reasonable because the information passed for the hidden states of CDR-H3 sequence in Eq.~\ref{eq:memory} is closely influenced by this term while the 3D coordinates are directly passed on to the next iteration.
(6) The external distance prediction loss $\gL_\text{dist}$ is vital for docking, which we suspect the coordinates alone are insufficient to correctly recover the structure of the shadow paratope specifically during early iterations. 


\begin{table}[t!]
\vskip -0.1in
\centering
\caption{Ablations of different components in dyMEAN.}
\label{tab:ablation}
\scalebox{0.67}{
\begin{tabular}{cccccccc}
\toprule
\multirow{2}{*}{Model} & \multicolumn{3}{c}{Generation}                       &           & \multicolumn{3}{c}{Docking}                       \\ \cline{2-4} \cline{6-8} 
                       & AAR$\uparrow$    & TMscore$\uparrow$& LDDT$\uparrow$&           & CAAR$\uparrow$& RMSD$\downarrow$& DockQ$\uparrow$\\ \hline
\rowcolor{black!5}\multicolumn{8}{c}{CDR-H3 Design}                                                                                                             \\ \hline
dyMEAN                 & \textbf{43.65\%} & 0.9726          & \textbf{0.8454} &           & 28.11\%          & \textbf{\hphantom{0}8.11} & \textbf{0.409} \\
$T=2$                  & 43.57\%          & \textbf{0.9731} & 0.8411          &           & \textbf{29.23\%} & \hphantom{0}9.68       & 0.383          \\
$T=4$                  & 42.84\%          & 0.9725          & 0.8440          &           & 28.18\%          & \hphantom{0}8.65  & 0.393          \\
- full-atom            & 41.81\%          & 0.9730          & 0.7999          &           & 27.96\%          & 10.10         & 0.343          \\
- sharing\hphantom{-0} & 43.17\%          & 0.9718          & 0.8374          &           & 28.79\%          & \hphantom{0}9.46  & 0.356          \\
- $\vw_i \vw_j^\top$\hphantom{00}   & 39.29\%          & 0.9724          & 0.8408          &           & 25.87\%          & \hphantom{0}8.60  & 0.407          \\
- memory\hphantom{0}               & 40.01\%          & 0.9727          & 0.8444          &           & 24.37\%          & \hphantom{0}9.03  & 0.378          \\
- $\gL_\text{dist}$\hphantom{0000}    & 42.32\%          & 0.9715          & 0.8361          &           & 27.46\%          & \hphantom{0}9.07  & 0.393          \\ \hline
\rowcolor{black!5}\multicolumn{8}{c}{Complex Structure Prediction}                                                                                              \\ \hline
dyMEAN                 & -                & \textbf{0.9731} & \textbf{0.8673} &           & -                & \hphantom{0}9.05  & \textbf{0.452} \\
$T=2$                  & -                & 0.9716          & 0.8606          &           & -                & \hphantom{0}9.61  & 0.440          \\
$T=4$                  & -                & 0.9712          & 0.8628          &           & -                & \hphantom{0}9.62  & 0.441          \\
- full-atom            & -                & 0.9713          & 0.8111          &           & -                & \hphantom{0}9.98  & 0.424          \\
- sharing\hphantom{-0} & -                & 0.9709          & 0.8641          &           & -                & 11.27             & 0.429          \\
- $\vw_i \vw_j^\top$\hphantom{00}   & -                & 0.9725          & 0.8662          &           & -                & \hphantom{0}9.16  & 0.432          \\
- memory\hphantom{0}               & -                & 0.9711          & 0.8653          &           & -                & \textbf{\hphantom{0}8.89} & 0.447          \\
- $\gL_\text{dist}$\hphantom{0000}    & -                & 0.9706          & 0.8587          &           & -                & \hphantom{0}9.97  & 0.416          \\ \bottomrule
\end{tabular}}
\vskip -0.1in
\end{table}
\vspace{-0.1in}
\paragraph{Multiple CDRs Design and Full Antibody Design}
\label{sec:designall}
In \textsection~\ref{sec:cdrh3}, we follow previous settings (e.g. HERN) and focuses mainly on the design of CDR-H3, because CDR-H3 is the loop mostly involved in binding and the most difficult to model. However, our method can be easily extended to include other CDRs or any target regions. What we need to do is just masking all of them to generate, while the overall generation architecture keeps the same. To illustrate this flexibility, we additionally extend our model to the simultaneous design of all 6 CDRs and report the results in Table~\ref{tab:designall}. It suggests that the results are promising in general.

Further, we provide the results for designing the full antibody including the framework regions in Table~\ref{tab:designall}. It reads that AAR improves by a large margin, which is expected because the other parts excluding CDR-H3 in an antibody exhibit stronger regularities and conservativeness. The side-chain generation (lDDT) worsens, which is also expected because it is more challenging to simultaneously generate the type of the residue as well as its side-chain geometry than generating the side chains with known residue type in the framework regions. The performance on backbone generation (TMscore) and docking (DockQ) remain similar to the CDR-H3 design experiment (\textsection~\ref{sec:cdrh3}).

\begin{table}[htbp]
\centering
\vskip -0.1in
\caption{Evaluation on designing all 6 CDRs simultaneously and designing the full antibody.}
\label{tab:designall}
\scalebox{0.9}{
\begin{tabular}{ccccc}
\toprule
\multicolumn{2}{c}{AAR} & TMscore              & lDDT   & DockQ   \\ \hline
\rowcolor{black!5}\multicolumn{5}{c}{Simultaneous Design of 6 CDRs}\\ \hline
\multicolumn{5}{c}{AAR details}                                   \\\hline
L1        & 75.55\%     &                      & H1     & 75.72\% \\
L2        & 83.10\%     &                      & H2     & 68.48\% \\
L3        & 52.12\%     & \multicolumn{1}{l}{} & H3     & 37.51\% \\ \hline
All       & 60.07\%     & 0.9653               & 0.8029 & 0.396   \\ \hline
\rowcolor{black!5}\multicolumn{5}{c}{Design of Full Antibody}     \\ \hline
Full      & 74.96\%     & 0.9662               & 0.7589 & 0.412   \\ \bottomrule
\end{tabular}}
\vskip -0.1in
\end{table}

\section{Limitations}

\paragraph{Data Diversity and Evaluation Metrics}
Currently, deep generative models are likely to face difficulties in antibody design due to the limited diversity of antigen-antibody data. We count the most frequent unigram pattern of the amino acid types of each position in CDR-H3 from the training set, matching from both sides to the middle, which yields the pattern $\mathrm{ARDG***DY}$ where most $\mathrm{*}$ are $\mathrm{Y}$. We use this unigram pattern to calculate AAR on the test set and obtain AAR$=39.61\%$ and CAAR$=26.57\%$. This implies that the meaningless unigram pattern is prevailing in both the training set and the test set, which may hinder the models from learning meaningful antigen-antibody interaction patterns and trick the evaluation metrics. After removing the first 4 residues and the last 2 residues from CDR-H3, dyMEAN achieves an AAR of $31.76\%$, which exhibits a clear performance detriment compared to Table~\ref{tab:cdr_gen}. These phenomena encourage future work in augmenting the dataset (\emph{e.g.} through wet-lab experiments or extracting similar interfaces from general protein complexes) and proposing better evaluation metrics to avoid the impact of the unigram distribution (\emph{e.g.} drop the residues that can be predicted correctly by unigram pattern).
\vspace{-0.2in}
\paragraph{Reliability of Computational Energy Functions}
Ultimately the binding affinity  (or binding energy) determines whether the generated candidates are good binders or not. In this paper, we use the deep-learning-based predictor of $\Delta\Delta G$, and it is also common to use statistical energy terms (\emph{e.g.} FoldX~\citep{schymkowitz2005foldx}, Rosetta~\citep{alford2017rosetta}, docking scores in softwares~\citep{goodsell1996automated}). However, the reliability of current computational energy functions still remains uncertain, and some are known to correlate poorly with the experimental results~\citep{ramirez2016reliable, ramirez2018reliable}. Indeed, there are two questions required to answer: (1) If these energy functions which are fitted on well binding complexes can distinguish poorly binding complexes? (2) If these energy functions which are fitted on natural complexes can generalize to complexes generated by deep models which may yield distinct distributions? We believe a well-generalizable affinity predictor is essential for learning-based antibody design; otherwise, wet-lab evaluations are necessary, which, yet, are inefficient and labor-intensive.

\section{Conclusion}
\label{sec:conclusion}
In this paper, we propose dyMEAN, a full-atom model for end-to-end antibody design given the epitope and the incomplete antibody sequence. Specifically, we explore a knowledge-guided structural initialization and propose shadow paratope for E(3)-equivariant message passing and docking. The proposed adaptive multi-channel encoder also tackles the challenge of the variant number of atoms in different residues in full-atom modeling. Our dyMEAN surpasses state-of-the-art models in terms of epitope-binding CDR-H3 design, complex structure prediction, and affinity optimization. Our work provides insights into the end-to-end antibody design and could inspire future research on the full-atom modeling of proteins.

\vspace{-0.1in}
\section*{Reproducibility}
Codes for our dyMEAN are available at \url{https://github.com/THUNLP-MT/dyMEAN}.

\vspace{-0.1in}
\section*{Acknowledgments}
This work is jointly supported by the Vanke Special Fund for Public Health and Health Discipline Development of Tsinghua University, the National Natural Science Foundation of China (No. 61925601, No. 62006137), Guoqiang Research Institute General Project of Tsinghua University (No. 2021GQG1012), Beijing Academy of Artificial Intelligence, Beijing Outstanding Young Scientist Program (No. BJJWZYJH012019100020098), Scientific Research Fund Project of Renmin University of China (Start-up Fund Project for New Teachers).

\bibliography{references}
\bibliographystyle{icml2023}

\newpage
\appendix
\onecolumn
\setcounter{theorem}{0}
\setcounter{equation}{0}

\section{Details of the Structural Initialization}
\label{app:init}

Given an antibody, we denote the position number of the $i$-th residue as $r_i$. With the backbone template $\{\mZ_{r} \in \R^{3 \times 4}| r \in \sW\}$ (\textsection~\ref{sec:init}) of the same numbering system, we initialize the backbone coordinates $\mZ_{r_i}$ of the structure by linearly interpolating the residues between the well-conserved ones and extending the residues at two ends outwards. The above process can be formalized as follows:
\begin{spacing}{1.5}
\vskip -0.2in
\begin{equation}
\vcenter{\hbox{\begin{minipage}{.3\textwidth}
\centering
\includegraphics[width=\textwidth]{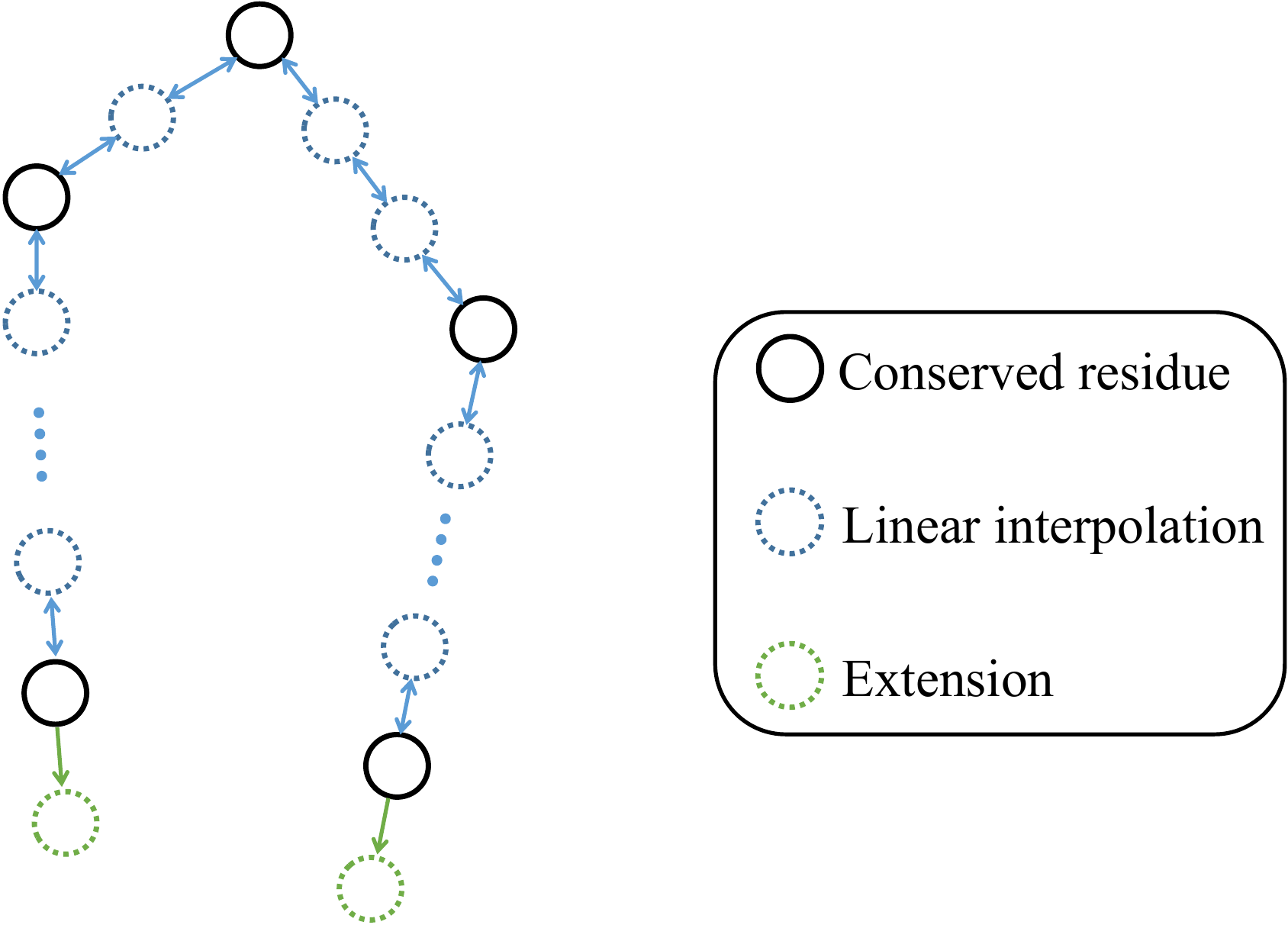}
\end{minipage}}}
\qquad\qquad
\begin{aligned}
    \mZ_{r_i} = \left\{
    \begin{array}{ll}
        \mZ_{r_i}, & r_i \in \sW, \\
        \textcolor[rgb]{0.25, 0.44, 0.61}{\frac{1}{q-p}[(i - p) \mZ_{r_q} + (q - i) \mZ_{r_p}],} & \textcolor[rgb]{0.25, 0.44, 0.61}{r_i \notin \sW, \exists p, q},\\
        \textcolor[rgb]{0.44, 0.68, 0.28}{\mZ_{r_p} + (i - p) (\mZ_{r_p} - \mX^{bb}_{p-1}),} & \textcolor[rgb]{0.44, 0.68, 0.28}{r_i \notin \sW, \exists p, \nexists q}, \\
        \textcolor[rgb]{0.44, 0.68, 0.28}{\mZ_{r_q} + (q - i) (\mZ_{r_q} - \mX^{bb}_{q+1}),} & \textcolor[rgb]{0.44, 0.68, 0.28}{r_i \notin \sW, \nexists p, \exists q}, 
    \end{array}
    \right.
\end{aligned}
\end{equation}
\vskip -0.2in
\end{spacing}
where $p$ and $q$ are defined as the indexes of the nearest conserved residues to the $i$-th residue when $r_i \notin \sP$:
\begin{align}
    p = \max_k \{k | k < i, r_k \in \sW\}, \\
    q = \min_k \{k | k > i, r_k \in \sW\},
\end{align}
$\mZ_{r_i}$ is then extended to $\mX_i^{(0)}$ by filling in the coordinates of the side-chain atoms with the coordinate of $\alpha$-carbon.

\section{Channel Attributes and Weights}
\label{app:channel_attr}
Given a multi-channel coordinate $\mX \in \R^{3\times c}$, we assign it with a $d$-dimensional attribute matrix $\mA \in \R^{c \times d}$ and a set of weights $\vw \in \R^{c \times 1}$. Each row vector of $\mA$ is associated with the atom type and the atom position of the corresponding channel and the weights are decided by the residue type of the coordinate. Natural amino acids only incorporate four atom types (i.e. $C, N, O, S$), and each atom will be assigned a position code indicating the number of chemical bonds on the shortest path from it to the $C_\alpha$~\citep{iupac1970abbreviations}. For example, Figure~\ref{fig:atompos} illustrates the position code of each atom in the side chain of Tryptophan. The attribute vector for each channel is the sum of its atom type embedding and its position code embedding. For the unknown residues, we assign them with a maximum number of atom channels, where each channel is filled with a $\mathrm{[MASK]}$ atom type and a $\mathrm{[MASK]}$ atom position.
\begin{figure}[htbp]
    \centering
    \includegraphics[width=.2\textwidth]{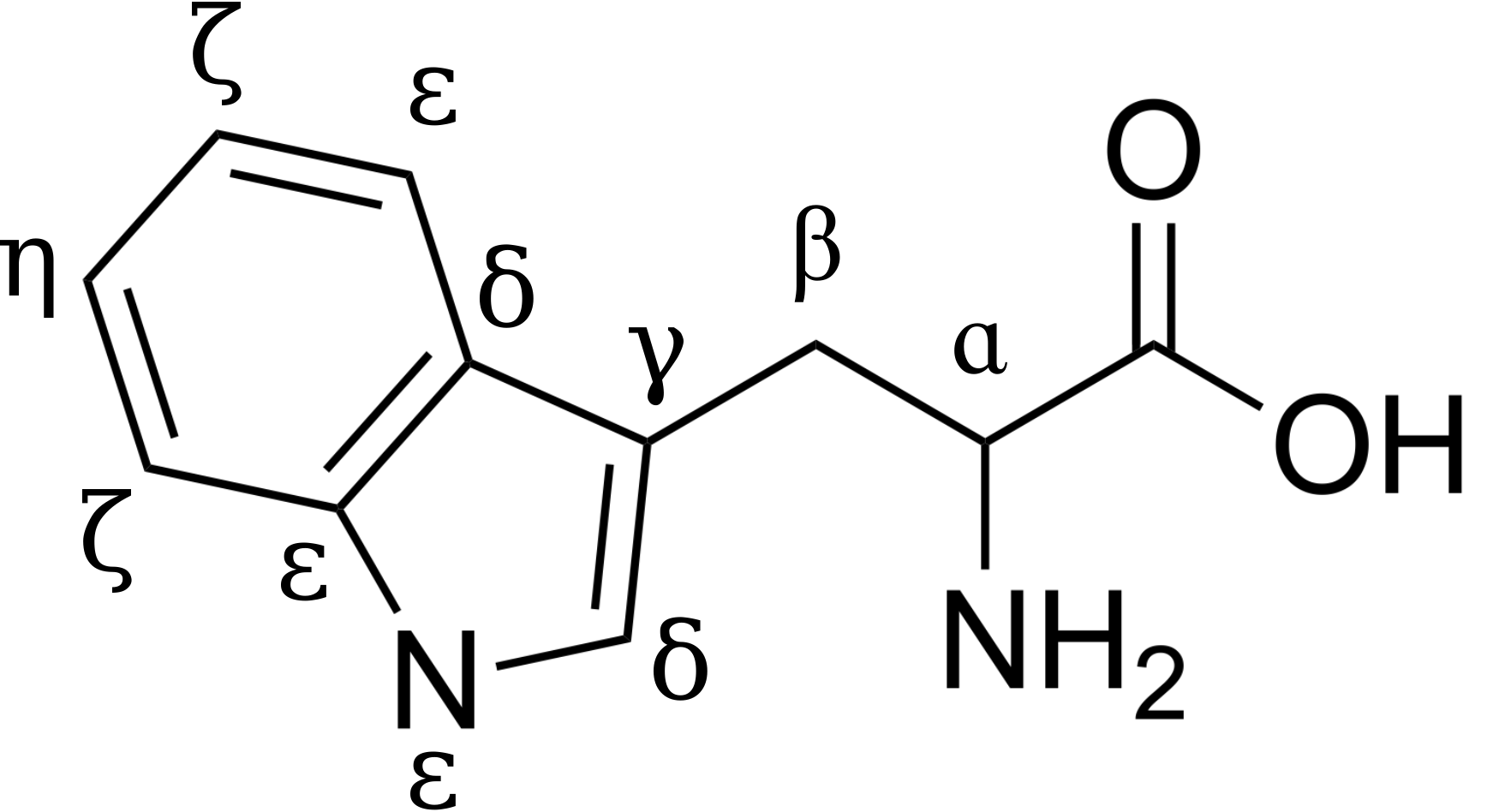}
    \caption{Position codes for atoms in the side-chain of Tryptophan.}
    \label{fig:atompos}
    \vskip -0.1in
\end{figure}

Both attribute vectors and weights are learnable parameters in our model. For efficiency consideration, the value of $d$ should be small because the dimension of the geometric relations $\mR_{ij} \in \R^{d\times d}$ (\textsection~\ref{sec:dyenc}) is quadratic of $d$. In practice, we find that $d = 16$ is sufficient. Furthermore, we normalize the weight $\vw$ with its L2-norm to avoid potential numerical instability.

\section{Proof of Theorem~\ref{the:e3}}
\label{app:e3}
\equivariance*

We start by giving the definition of E$(3)$-invariance and E$(3)$-equivariance~\citep{huang2022equivariant} as follows:
\begin{definition}[E$(3)$-equivariance]
    A function $\phi: \sX \mapsto \sY$ is E$(3)$-equivariant if $\forall g \in \text{E}(3)$ we have $\rho_\sY(g) \vy = \phi(\rho_\sX(g) \vx)$, where $\rho_\sX$ and $\rho_\sY$ instantiate $g$ in the input space $\sX$ and the output space $\sY$. Specifically, $\phi$ is E$(3)$-invariant if $\rho_\sY(g) \equiv I$, where $I$ is the identity transformation.
\end{definition}

Before proceeding to the overall proof, we need to first present and prove several necessary lemmas below.

\begin{lemma}
    \label{lem:tr}
    For the geometric relation extractor $T_R$ (\textsection~\ref{sec:dyenc}), $\forall \mX_i \in \R^{3 \times c_i}, \mX_j \in \R^{3 \times c_j}$, suppose $\mR_{ij} = T_R\left(\mX_i, \mX_j\right)$, then $T_R$ is E(3)-invariant. Namely, $\forall g \in \text{E}(3)$, we have $\mR_{ij} = T_R\left(g \cdot \mX_i, g \cdot \mX_j\right)$, where $g\cdot \mX \coloneqq \mQ \mX + \vt$, $\mQ \in \text{O}(3), \vt \in \R^3$.
\end{lemma}
\begin{proof}
$\forall \mX_i \in \R^{3 \times c_i}, \mX_j \in \R^{3 \times c_j}$, $\mR_{ij}$ is obtained through Eq.~\ref{eq:tr}. Consider the pair-wise channel distance matrix $\mD_{ij}$, $\forall \mQ \in \text{O}(3), \vt \in \R^3$, we have:
\begin{align*}
    \mD_{ij}(p, q) &= ||(\mQ\mX_i(:, p) + \vt) - (\mQ\mX_j(:, q) + \vt)|| = ||\mQ(\mX_i(:, p) - \mX_j(:, q))||, \\
                   &= \sqrt{[\mQ(\mX_i(:, p) - \mX_j(:, q))]^\top[\mQ(\mX_i(:, p) - \mX_j(:, q))]}, \\
                   &= \sqrt{(\mX_i(:, p) - \mX_j(:, q))^\top\mQ^\top\mQ(\mX_i(:, p) - \mX_j(:, q))}, \\
                   &= \sqrt{(\mX_i(:, p) - \mX_j(:, q))^\top(\mX_i(:, p) - \mX_j(:, q))}, \\
                   &= ||\mX_i(:, p) - \mX_j(:, q)||,
\end{align*}
With $\mA_i, \mA_j, \vw_i, \vw_j$ not affected by transformations on $\mX_i$ and $\mX_j$, we can directly derive $\mR_{ij} = T_R\left(\mQ\mX_i + \vt, \mQ\mX_j + \vt\right) = T_R\left(\mX_i, \mX_j\right)$, which concludes Lemma~\ref{lem:tr}.
\end{proof}

\begin{lemma}
    \label{lem:ts}
    For the geometric message scaler $T_S$ (\textsection~\ref{sec:dyenc}), $\forall \mX \in \R^{3 \times c}, \forall \vs \in \R^{C}$, suppose $\mX' = T_S\left(\mX, \vs\right)$, then $T_S$ is $\text{O}(3)$-equivariant. Namely, $\forall \mQ \in \text{O}(3)$, we have $\mQ \mX' = T_S\left(\mQ \mX, \vs\right)$.
\end{lemma}
\begin{proof}
    $\forall \mQ \in \text{O}(3), \vt \in \R^3$, we have:
    \begin{align*}
        T_S\left(\mQ \mX, \vs\right) = (\mQ \mX) \mathrm{diag}(\vs') = \mQ (\mX \mathrm{diag}(\vs') = \mQ T_S\left(\mX, \vs\right) = \mQ \mX',
    \end{align*}
    which concludes Lemma~\ref{lem:ts}.
\end{proof}

\begin{lemma}
    \label{lem:ame}
    Denote the AME (\textsection~\ref{sec:dyenc}) as $\{\vh_i', \mX_i'\}_{i\in \gV_E\cup\gV_S\cup\gV_A} = \text{AME}\left(\{\vh_i, \mX_i\}_{i\in \gV_E\cup\gV_S\cup\gV_A}\right)$, then AME is independent $\text{E}(3)$-equivariant~\citep{ganea2021independent} with respect to $\gV_E\cup\gV_S$ and $\gV_A$. Namely, $\forall g_1, g_2 \in \text{E}(3)$, we have $\{\vh_i', g_1\cdot\mX_i'\}_{i\in \gV_E\cup\gV_S}\cup \{\vh_i', g_2\cdot\mX_i'\}_{i\in\gV_A} = \text{AME}\left(\{\vh_i, g_1\cdot\mX_i\}_{i\in \gV_E\cup\gV_S}\cup \{\vh_i, g_2\cdot\mX_i\}_{i\in\gV_A}\right)$.
\end{lemma}
\begin{proof}
    The key points to the proof of Lemma~\ref{lem:ame} are the following two statements: (1) The information exchange between $\gV_E\cup\gV_S$ and $\gV_A$ is $\text{E}(3)$-invariant; (2) The propagation process of Eq.~\ref{eq:mp1}-\ref{eq:mp4} is $\text{E}$(3)-invariant on $\vh$ and $\text{E}$(3)-equivariant on $\mX$. If both (1) and (2) are right, then each layer of the AME satisfies independent $\text{E}(3)$-equivariance with respect to $\gV_E\cup\gV_S$ and $\gV_A$, which easily leads to the correctness of Lemma~\ref{lem:ame}. Suppose (2) holds true, then the correctness of (1) is obvious because the information exchange between $\gV_E\cup\gV_S$ and $\gV_A$ is conducted by the sharing of hidden states and topology between $\gV_S$ and $\gV_P$, both of which are $\text{E}(3)$-invariant. Thus the focus narrows down to the proof of (2).

    $\forall g \in \text{E}(3)$, $g\cdot \mX\coloneqq\mQ \mX + \vt$, $\mQ \in \text{O}(3), \vt \in \R^3$, according to Lemma~\ref{lem:tr}, we have:
    \begin{align*}
        \vm_{ij} &= \phi_m(\vh_i^{(l)}, \vh_j^{(l)}, \frac{T_R(\mX_i^{(l)}, \mX_j^{(l)})}{||T_R(\mX_i^{(l)}, \mX_j^{(l)})||_F + \epsilon}), \\
                 &= \phi_m(\vh_i^{(l)}, \vh_j^{(l)}, \frac{T_R(g\cdot\mX_i^{(l)}, g\cdot\mX_j^{(l)})}{||T_R(g\cdot\mX_i^{(l)}, g\cdot\mX_j^{(l)})||_F + \epsilon}),
    \end{align*}
    which reads that the computation of $\vm_{ij}$ is $\text{E}(3)$-invariant. This directly leads to the $\text{E}(3)$-invariance of obtaining $\vh_i^{(l+1)}$ through Eq.~\ref{eq:mp3}. Next, according to Lemma~\ref{lem:ts}, we have:
    \begin{align*}
        \mQ \mX_{ij} &= \mQ T_S(\mX_{i}^{(l)} - \frac{1}{c_j} \sum_{k = 1}^{c_j} \mX_j^{(l)}(:, k), \phi_x(\vm_{ij})), \\
                     &= T_S(\mQ (\mX_{i}^{(l)} - \frac{1}{c_j} \sum_{k = 1}^{c_j} \mX_j^{(l)}(:, k)), \phi_x(\vm_{ij})), \\
                     &= T_S(\mQ \mX_{i}^{(l)} - \frac{1}{c_j} \sum_{k = 1}^{c_j} \mQ\mX_j^{(l)}(:, k), \phi_x(\vm_{ij})), \\
                     &= T_S(\mQ \mX_{i}^{(l)} + \vt - \frac{1}{c_j} \sum_{k = 1}^{c_j} (\mQ\mX_j^{(l)}(:, k) + \vt), \phi_x(\vm_{ij})), \\
                     &= T_S(g\cdot\mX_{i}^{(l)} - \frac{1}{c_j} \sum_{k = 1}^{c_j} g\cdot\mX_j^{(l)}(:, k), \phi_x(\vm_{ij})),
    \end{align*}
    which leads to the $\text{E}(3)$-equivariance of obtaining $\mX_i^{(l+1)}$ through Eq.~\ref{eq:mp4}:
    \begin{align*}
        g\cdot \mX_i^{(l+1)} &= g\cdot (\mX_i^{(l)} + \frac{1}{|\gN(i)|} \sum\nolimits_{j \in \gN(i)} \mX_{ij}), \\
                             &= \mQ (\mX_i^{(l)} + \frac{1}{|\gN(i)|} \sum\nolimits_{j \in \gN(i)} \mX_{ij})) + \vt, \\
                             &= \mQ \mX_i^{(l)} + \vt + \mQ \frac{1}{|\gN(i)|} \sum\nolimits_{j \in \gN(i)} \mX_{ij}, \\
                             &= g\cdot \mX_i^{(l)} + \frac{1}{|\gN(i)|} \sum\nolimits_{j \in \gN(i)} \mQ \mX_{ij}
    \end{align*}
    Therefore, the propagation process of Eq.~\ref{eq:mp1}-\ref{eq:mp4} is $\text{E}$(3)-invariant on $\vh$ and $\text{E}$(3)-equivariant on $\mX$, which concludes Lemma~\ref{lem:ame}.
\end{proof}

\begin{lemma}
    \label{lem:dock}
    Denote the docking procedure in Eq.~\ref{eq:dock1}-\ref{eq:dock2} as $\{\tilde{\mX}_i\}_{i\in\gV_A} = \mathrm{Dock}(\{\mX_i^{(T)}\}_{i\in\gV_A}, \{\mX_i^{(T)}\}_{i\in\gV_S})$, then it is E$(3)$-equivariant in terms of $\gV_S$. Namely, $\forall g_1, g_2 \in \text{E}(3)$, we have $\{g_1 \cdot \tilde{\mX}_i\}_{i\in\gV_A} = \mathrm{Dock}(\{g_2 \cdot \mX_i^{(T)}\}_{i\in\gV_A}, \{g_1 \cdot \mX_i^{(T)}\}_{i\in\gV_S})$.
\end{lemma}

\begin{proof}
    Suppose $g\cdot X \coloneqq \mQ \mX + \vt$, where $\mQ, \vt = \mathrm{Kabsch}(\{\mX_i^{(T)}\}_{i\in\gV_P}, \{\mX_i^{(T)}\}_{i\in\gV_S})$. When we apply $g_1$ to $\gV_S$ and $g_2$ to $\gV_A$ ($\gV_P \in \gV_A$), the new Kabsch process can be interpreted as first exerting $g_1^{-1}$ on $\gV_S$ and $g_2^{-1}$ on $\gV_A$ to eliminate the transformations, then implementing the above-mentioned $g$ on $\gV_A$, and finally transforming $\gV_A$ with $g_1$ to recover the transformation on $\gV_S$. Therefore, for $g'\cdot X \coloneqq \mQ' \mX + \vt'$, where $\mQ', \vt' = \mathrm{Kabsch}(\{g_2 \cdot \mX_i^{(T)}\}_{i\in\gV_P}, \{g_1 \cdot \mX_i^{(T)}\}_{i\in\gV_S})$, we have $g' = g_1 \cdot g\cdot g_2^{-1}$. Then it is easy to derive:
    \begin{align*}
        \{g' \cdot g_2 \cdot \mX_i^{(T)}\}_{i\in\gV_A}
        &= \mathrm{Dock}(\{g_2 \cdot \mX_i^{(T)}\}_{i\in\gV_A}, \{g_1 \cdot \mX_i^{(T)}\}_{i\in\gV_S}) \\
        &= \{g_1 \cdot g\cdot g_2^{-1} \cdot g_2 \cdot \mX_i^{(T)}\}_{i\in\gV_A} = \{g_1 \cdot g \cdot \mX_i^{(T)}\}_{i\in\gV_A} \\
        &= \{g_1 \cdot \tilde{\mX}_i\}_{i\in\gV_A}, 
    \end{align*}
    which concludes Lemma~\ref{lem:dock}.
\end{proof}

Ultimately we are ready to give the proof to Theorem~\ref{the:e3} as follows:
\begin{proof}
    $\forall g_1, g_2 \in \text{E}(3)$, according to Lemma~\ref{lem:ame}, each iteration in dyMEAN satisfies independent $\text{E}(3)$-equivariance with respect to $\gV_E\cup\gV_S$ and $\gV_A$. Thus the transformed inputs $\{\vh_i,g_1\cdot\mX_i\}_{i\in\gV_E\cup\gV_S}, \{\vh_i^{(0)},g_2\cdot\mX_i^{(0)}\}_{i\in\gV_A}$ lead to transformed encoded results with independent $\text{E}(3)$-equivariance $\{\vh_i^{(T)},g_1\cdot\mX_i^{(T)}\}_{i\in\gV_S}, \{\vh_i^{(T)},g_2\cdot\mX_i^{(T)}\}_{i\in\gV_A}$. Next, based on Lemma~\ref{lem:dock}, the docking procedure is E$(3)$-equivariant in terms of $\gV_S$, thus we have $\{g_1 \cdot \tilde{\mX}_i\}_{i\in\gV_A} = \mathrm{Dock}(\{g_2 \cdot \mX_i^{(T)}\}_{i\in\gV_A}, \{g_1 \cdot \mX_i^{(T)}\}_{i\in\gV_S})$.
    Plus that $\vp_i$ is obtained through Softmax on $\vh_i^{(T)}$, we can derive:
    \begin{align*}
        \{\vp_i\}_{i\in\gV_P}, \{g_1\cdot\tilde{\mX_i}\}_{i\in\gV_A}=\text{dyMEAN}\left(\{\vh_i,g_1\cdot\mX_i\}_{i\in\gV_E\cup\gV_S}, \{\vh_i^{(0)},g_2\cdot\mX_i^{(0)}\}_{i\in\gV_A}\right),
    \end{align*}
    which concludes Theorem~\ref{the:e3}. Also, the initialization of the shadow paratope conforms to the standard Gaussian distribution at the center of the epitope, which is isotropic~\citep{xu2022geodiff} and thus does not interfere with the $\text{E}(3)$-equivariance.
\end{proof}

\section{Adaption for Property Optimization}
\label{app:prop_opt}
We can also adjust our method for optimizing the properties (e.g. affinity) of existing antibodies. (1) \textbf{Initialization with disturbance}. Since the structure of the existing antigen-antibody complex is known, we only need to disturb it with $\tZ\sim \gN(0, \tI)$ for initialization rather than the method in \textsection~\ref{sec:init}. (2) \textbf{No shadow paratope}. The initialization also provides the relative position of the antibody to the epitope, therefore we can directly capture the interface geometry through interacting edges between epitope and antibody without the shadow paratope in \textsection~\ref{sec:shadow}. A pair of residues is interacting if their distance is below a threshold $\delta$ (i.e. 6.6 \mbox{\normalfont\AA} according to \citet{ramaraj2012antigen}), therefore the interacting edge set is $\gE_\gI = \{(u, v)|d(u, v) \leq \delta, u \in \gV_E, v \in \gV_A\}$. The alternating message passing in the encoder AME can be merged into a single step on the entire complex $(\gV_E \cup \gV_A, \gE_E \cup \gE_A \cup \gE_\gI)$.
(3) \textbf{Partially masked sequence}. A common requirement of property optimization is to change the sequence as little as possible~\citep{ren2022proximal}. To achieve flexibility in controlling the extent of modification, we only mask a random subset of the CDR residues $\sI^\prime \subseteq \sI$ for each training step. Then we can set the upper bound of the number of modified residues by masking a fixed number of residues for generation.

\paragraph{Gradient Search} Given a property scorer on complexes $f: \sG \rightarrow \R$, we search for a favorable initialization that leads to a complex with the better property. we first utilize our trained model $\theta$ to generate a dataset $\gD = \{(\vh_{\gG_i}, f(\gG_i))|\gG_i \sim \theta\}$, where $\vh_{\gG_i}$ denotes the representation of the complex $\gG_i$ obtained by averaging the hidden states of all nodes. The dataset $\gD$ is then used to fit a predictor of $f$ on the representation space of complexes: $\hat{f}_\theta: \sH \rightarrow f[\sG]$. Now given an existing complex $\gG$ with CDRs partially masked, the mapping from the initial disturbance to the predicted score $\hat{f}_\theta \circ g_\theta (\gG | \tZ)$ is differentiable, where $g_\theta$ sums up the process of generating a complex and then obtaining its representation. Suppose we want to maximize the property score, we can conduct gradient search on the initialization space by minimizing the target function:
\begin{align}
    \gL(\tZ) = -\hat{f}_\theta \circ g_\theta (\gG | \tZ) + \KL (\tZ \Vert \gN(0, \tI)),
\end{align}
where the KL divergence~\citep{kullback1951information} restricts the disturbance to the standard Gaussian distribution.

\section{Huber Loss}
\label{app:huber}
The Huber loss~\citep{huber1992robust} for robust supervision of the coordinates and bond lengths is defined as follows: 
\begin{align}
    l(x, y) = \left\{ \begin{array}{lr}
    0.5 (x-y)^2, if |x-y| < \delta, \\
    \delta \cdot (|x-y| - 0.5 \cdot \delta), else
\end{array} \right.
\end{align}
When the deviation between $x$ and $y$ is below the threshold $\delta$, $l$ is equivalent to MSE loss, and when the deviation is above the threshold, $l$ is equivalent to L1 loss. MSE loss provides smoothness near 0 but is sensitive to outliers, while the opposite holds true for L1 loss. By selecting a suitable loss according to the deviation, Huber loss combines the merits of MSE loss and L1 loss, thus exhibiting better numerical stability. We set $\delta = 1$ in our experiments following~\citet{kong2022conditional}.

\section{Experiment Details}
\label{app:exp}
For the baselines, we adopt the hyperparameters and training procedure in their official releases since all the papers utilize SAbDab to form training sets of similar scale and distribution. We list the values of these hyperparameters as well as those of our dyMEAN in Table~\ref{tab:hyperparam}.

We train dyMEAN by Adam optimizer with the data-parallel framework of PyTorch on 2 GeForce RTX 2080 Ti GPUs. We set the initial $lr = 1 \times 10^{-3}$ and decay the learning rate exponentially to reach $1\times 10^{-4}$ at the last step. The batch size is 16, which is consistent across different tasks. It takes 200 epochs for dyMEAN to converge in the tasks of epitope-binding CDR-H3 design and affinity optimization, while the number is 250 in the task of complex structure prediction. We notice that the learning of 1D sequences is faster than that of 3D structures, leading to overfitting the 1D sequences in CDR-H3 design. To bypass the problem, we unmask some paratope residues at the initial stage of training, and gradually transit to the ultimate setting where all the paratope residues are masked. Specifically, the ratio of unmasked paratope residues is initialized as $90\%$ and anneals to $0\%$ with a cosine schedule.

\begin{table}[htbp]
\centering
\caption{Hyperparameters for the baselines and our dyMEAN.}
\label{tab:hyperparam}
\begin{tabular}{ccl}
\toprule
hyperparameter       & value                & description                                                                        \\ \hline
\rowcolor{black!5}\multicolumn{3}{c}{HERN}                                                                                                         \\ \hline
hidden\_size         & 256                  & Size of the hidden states in its hierachical message passing network (MPN).        \\
num\_rbf             & 16                   & Number of RBF kernels for distance embedding.                                      \\
n\_layers            & 4                    & Number of layers in the MPN.                                                       \\
k\_neighbors         & 9                    & Number of neighbors for each node in the KNN graph.                                \\ \hline
\rowcolor{black!5}\multicolumn{3}{c}{Diffab}                                                                                                       \\ \hline
hidden\_size         & 128                  & Size of the hidden states in the MPN.                                              \\
pair\_size           & 64                   & Size of the residue-pair features.                                                 \\
n\_layers            & 6                    & Number of layers in the MPN.                                                       \\
n\_steps             & 100                  & Number of the diffusion steps.                                                     \\ \hline
\rowcolor{black!5}\multicolumn{3}{c}{MEAN}                                                                                                         \\ \hline
embed\_size          & 64                   & Size of the residue type embedding.                                                \\
hidden\_size         & 128                  & Size of the hidden states in the MPN                                               \\
n\_layers            & 3                    & Number of layers in the MPN                                                        \\
n\_iter              & 3                    & Number of iterations in its progressive full-shot decoding.                        \\
k\_neighbors         & 9                    & Number of neighbors for each node in the KNN graph.                                \\ \hline
\rowcolor{black!5}\multicolumn{3}{c}{dyMEAN (ours)}                                                                                                \\ \hline
embed\_size          & 64                   & Size of the residue type embedding and the position number embedding.              \\
hidden\_size         & 128                  & Size of the hidden states in the MPN                                               \\
n\_layers            & 3                    & Number of layers in the MPN                                                        \\
n\_iter              & 3                    & Number of iterations in the progressive full-shot decoding.                        \\
k\_neighbors         & 9                    & Number of neighbors for each node in the KNN graph.                                \\
$d$                  & 16                   & Size of the attribute vector of each channel (equal to the size of the atom type\\
\multicolumn{1}{l}{} & \multicolumn{1}{l}{} & embedding and the atom position embedding).                                                  \\ \bottomrule
\end{tabular}
\end{table}

Furthermore, we provide the definition of the epitope to HDock when using it for docking. To further enhance its docking performance, we generate 100 docked samples for each antibody and calculate the 48 residues closest to each antibody. We compare these residues with the given epitope and select the candidate with the top-1 coverage as the final result.

\section{Space Complexity Analysis}
\label{app:complexity}
We emphasize the spatial efficiency of our dyMEAN compared to HERN~\citep{jin2022antibody}, which also models the side chains in addition to backbones but is limited to the paratope (\emph{i.e.}, CDR-H3), by two aspects: initialization and encoding. We denote the number of residues in the epitope, the paratope, and the antibody by $N_E, N_P$, and $N_A$.

\paragraph{Initialization}
It is obvious that our structural initialization (\textsection~\ref{sec:init}) has a complexity linear to the number of residues in the shadow paratope and the antibody, which is $O(N_P + N_A)$. HERN initializes the coordinates of the antibody via eigenvalue decomposition of the residue-level pair-wise distance matrix of the complex, thus having a complexity of $O((N_E + N_P)^2)$. Since $N_A$ is much larger than $N_P$, the quadratic complexity largely impedes HERN's scaling from modeling paratope only to modeling the entire antibody (\emph{i.e.}, replacing $N_P$ with $N_A$).

\paragraph{Encoding}
The space complexity of GNN-based encoders is dominated by the scheme of edge-wise message passing. We denote the maximum number of neighbors of each node by $K$, and the maximum number of atoms in a single residue (\emph{i.e.}, maximum channel size) by $C$. For our AME (\textsection~\ref{sec:dyenc}), the major influence is the geometric relation extractor (Eq.~\ref{eq:tr}) with a complexity of $O(K(N_E + N_P + N_A)(2dC + 2C + C^2)) = O(K(N_E + N_P + N_A)C(2d + 2 + C))$, where $d$ is the dimension of the attribute vector. HERN adopts a hierarchical encoder, which first implements EGNN~\citep{satorras2021n} on the residue-level graph and the atom-level graph sequentially, then updates the coordinates with inter-$C_\alpha$ terms and intra-residue terms. Since the atom-level graph is much larger than its residue-level counterpart, the dominant part is the atom-level message passing with a complexity of $O(K(N_E + N_P)CH)$, where $H$ denotes the hidden size. Scaling HERN from $N_P$ to $N_A$, we have $N_E + N_P + N_A \approx N_E + N_A$ but $H >> 2d + 2 + C$ because $d$ is set to a small number and $C=14$ in the dataset, which reveals the superiority of AME in efficiency over the hierarchical encoder.

Our attempts to scale HERN to the entire antibody bring no success and exhibit unrealistic GPU requirements, which we attribute to the high complexity of its initialization, hierarchical encoder, and autoregressive refinement~\citep{kong2022conditional}.

\section{Side-Chain Dihedral Angles}
\label{app:xangles}
To analyze whether our model generates realistic dihedral angles in the side chains, we draw the distribution of $\chi_1, \chi_2, \chi_3, \chi_4$ with the generated structures and the reference structures\footnote{We use the following definitions of the dihedral angles: \url{http://www.mlb.co.jp/linux/science/garlic/doc/commands/dihedrals.html}}.
We display the overall distribution in figure~\ref{fig:chis_all} and separated distribution of different amino acids in figure~\ref{fig:chis_separate}. Both figure~\ref{fig:chis_all} and figure~\ref{fig:chis_separate} show that generally the generated dihedral angles conforms to the reference distribution. However, we also identify from the fine-grained figure~\ref{fig:chis_separate} that the generated distributions are smoother than the reference ones, indicating possible minor deviations on the angles under certain circumstances. Therefore, in some practical applications, relaxing methods like OpenMM~\citep{eastman2017openmm} are still needed for post-process.
\begin{figure}[htbp]
    \vskip -0.1in
    \centering
    \includegraphics[width=\textwidth]{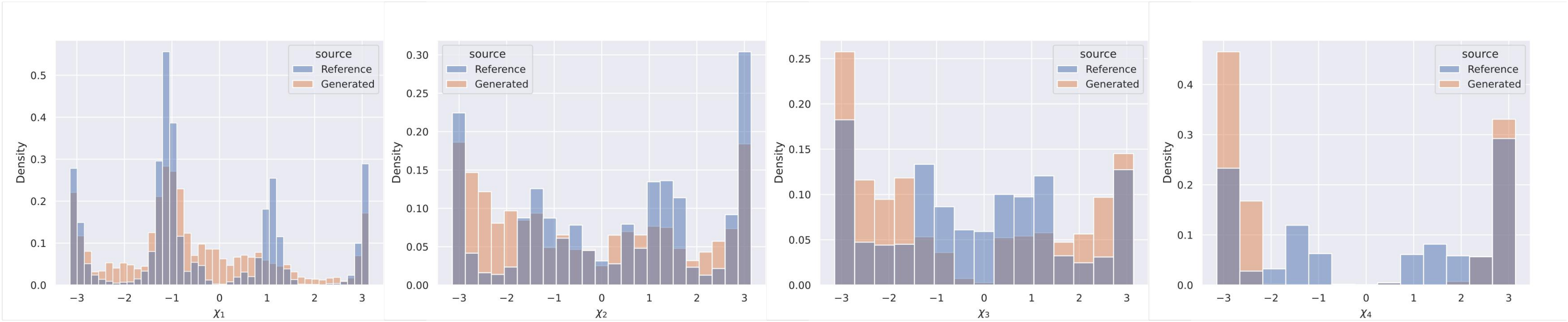}
    \vskip -0.1in
    \caption{The overall distributions of 4 dihedral angles in the side chains.}
    \label{fig:chis_all}
\end{figure}

\begin{figure}[htbp]
    \vskip -0.2in
    \centering
    \includegraphics[width=.9\textwidth]{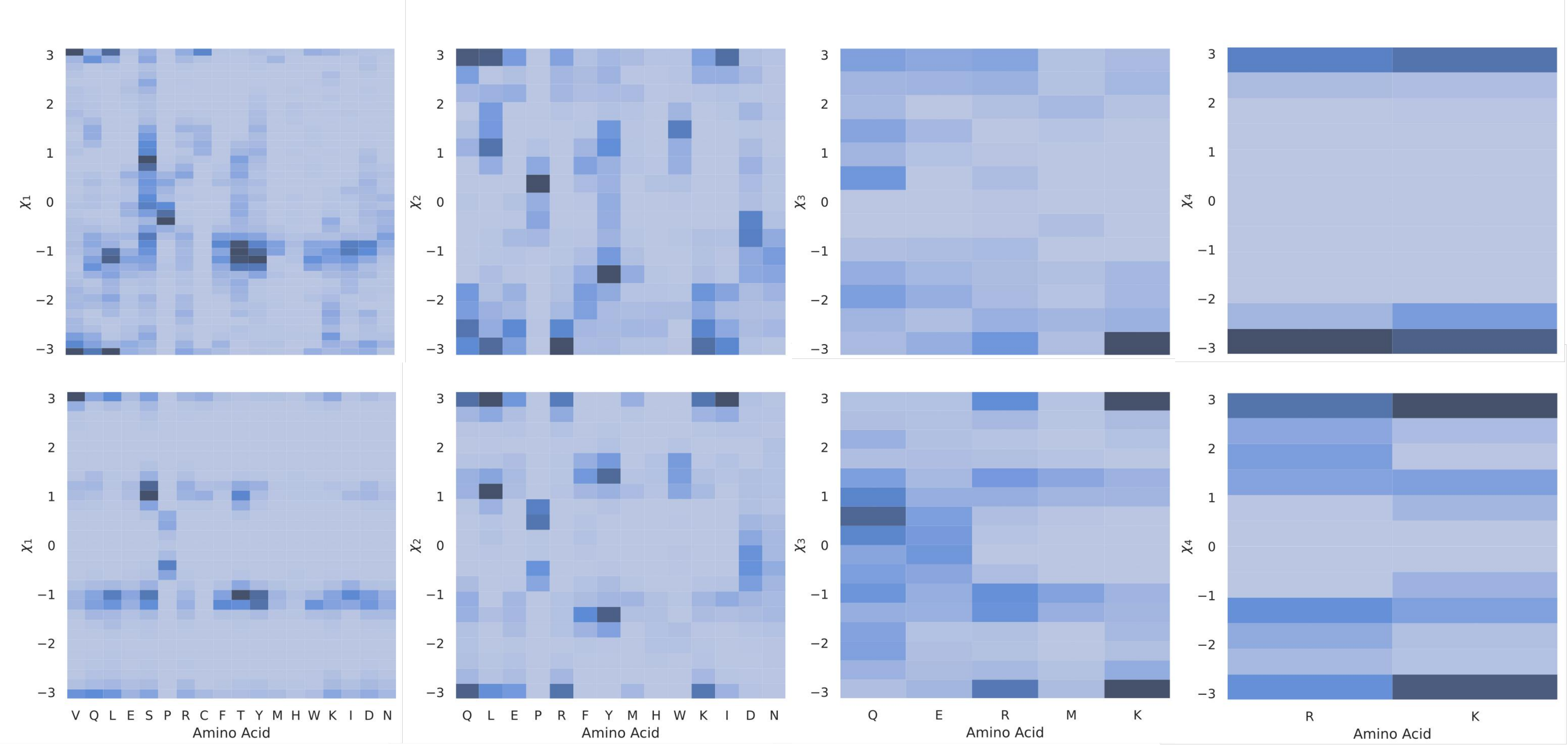}
    \vskip -0.1in
    \caption{The distribution of 4 dihedral angles in the side chains categorized by amino acids. The first row and the second row display the distributions from the generated structures and the reference structures, respectively.}
    \label{fig:chis_separate}
\end{figure}

\section{Threshold for Defining Conserved Residues}
\label{app:conserve_th}
We analyze the influence of the threshold for defining the conserved residues (\textsection~\ref{sec:init}) by displaying the variations in the number of conserved residues and the average RMSD of the antibodies in the dataset to the conserved template. According to Figure~\ref{fig:conserve_th}, with the decrease of the threshold, the number of conserved residues gradually increases and becomes stable at 90\% threshold, which is expected. In contrast, the RMSD curve shows a surge at 92\% threshold. Hence, to balance these two factors, we think it is better to set the threshold between 93\% and 96\%, where the RMSD remains low and the number of conserved residues is not too small. 

\begin{figure}[htbp]
    \centering
    \includegraphics[width=.7\textwidth]{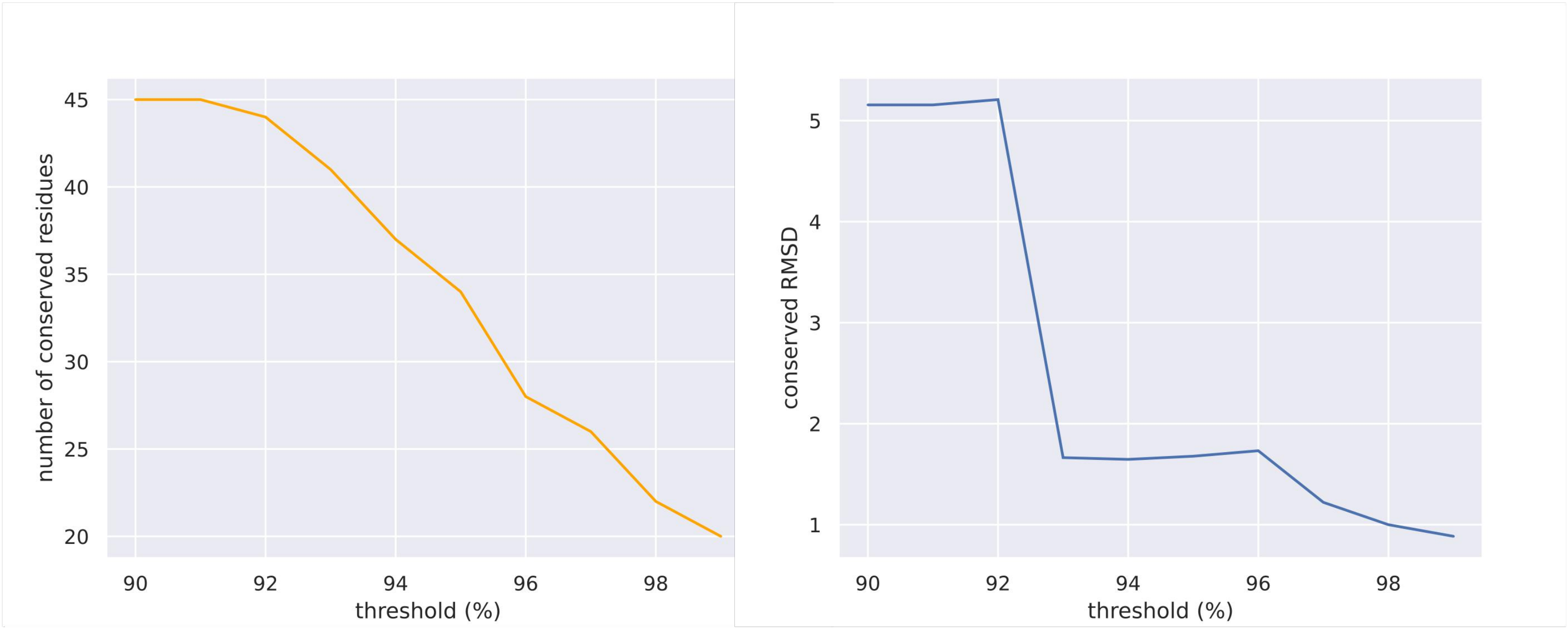}
    \caption{Number of conserved residues (left) and average RMSD of the antibodies in the dataset to the conserved template (right) with respect to different thresholds for defining the conserved residues.}
    \label{fig:conserve_th}
\end{figure}

We further conducted experiments with templates from 90\% and 99\% threshold, and show the results in Table~\ref{tab:conserve_th}. It is observed that our model is robust with the templates, but still, using 95\% as threshold generally achieves more favorable results compared to the other choices.

\begin{table}[htbp]
\centering
\caption{Performance of dyMEAN with conserved templates defined by different thresholds.}
\label{tab:conserve_th}
\scalebox{0.9}{
\begin{tabular}{cccccccc}
\toprule
\multirow{2}{*}{Threshold} & \multicolumn{3}{c}{Generation}                       &           & \multicolumn{3}{c}{Docking}                       \\ \cline{2-4} \cline{6-8} 
                       & AAR$\uparrow$    & TMscore$\uparrow$& LDDT$\uparrow$&           & CAAR$\uparrow$& RMSD$\downarrow$& DockQ$\uparrow$\\ \hline
\rowcolor{black!5}\multicolumn{8}{c}{CDR-H3 Design}                                                                                                             \\ \hline
90\%                   & 40.99\%          & 0.9722          & 0.8365          &           & 26.04\%          & \textbf{\hphantom{0}7.88} & \textbf{0.415} \\
95\%                   & \textbf{43.65\%} & 0.9726          & \textbf{0.8454} &           & 28.11\%          & \hphantom{0}8.11          & 0.409          \\
99\%                   & 43.32\%          & \textbf{0.9730} & 0.8438          &           & 27.59\%          & \hphantom{0}9.30          & 0.405          \\
\hline
\rowcolor{black!5}\multicolumn{8}{c}{Complex Structure Prediction}                                                                                              \\ \hline
90\%                   & -                & 0.9703          & 0.8597          &           & -                & \hphantom{0}9.34           & \textbf{0.459} \\
95\%                   & -                & \textbf{0.9731} & \textbf{0.8673} &           & -                & \textbf{\hphantom{0}9.05}  & 0.452          \\
99\%                   & -                & 0.9712          & 0.8601          &           & -                & \hphantom{0}9.70           & 0.443          \\
\bottomrule
\end{tabular}}
\end{table}

\section{Selection of the Shadow Paratope in Docking}
\label{app:mulcdr_dock}
\begin{table}[htbp]
    \centering
    \caption{Results on complex structure prediction using both CDR-H3 and CDR-L3 as the shadow paratope.}
    \label{tab:mulcdr_dock}
    \begin{tabular}{lcccc}
        \toprule
        Shadow Paratope & TMscore$\uparrow$ & lDDT$\uparrow$    & RMSD$\downarrow$  & DockQ$\uparrow$ \\\midrule
        H3              & 0.9731            & 0.8673            & \hphantom{0}9.05  & 0.452           \\
        H3 + L3         & 0.9585            & 0.8248            & 10.80             & 0.397           \\\bottomrule
    \end{tabular}
\end{table}

In this work, we generally use CDR-H3 since the interacting residues mainly come from it~\citep{akbar2021compact}. In practice, we find that using CDR-H3 only is generally sufficient for docking. Also, it is easy to extend the shadow paratope to contain other CDRs by our implementation. We conduct an experiment using both CDR-H3 and CDR-L3 as the shadow paratope for the docking experiment and the results are in Table~\ref{tab:mulcdr_dock}. Additional inclusion of CDR-L3 leads to a slight performance drop in all metrics. This is possibly because the two CDRs are separated spatially in the structure, and identifying their relative positions brings up additional complexity to the problem.

\section{Affinity with FoldX}
\label{app:aff_foldx}
We provide the evaluation on affinity optimization (\textsection~\ref{sec:affopt}) with FoldX~\citep{schymkowitz2005foldx} as the criterion in Table~\ref{tab:aff_opt_foldx}. We first use OpenMM~\citep{eastman2017openmm} to relax the generated structures, then use FoldX to minimize the energy. Finally, we use FoldX to calculate the interface energy ($\Delta G$) of both the wild type and the mutant to obtain the $\Delta\Delta G$.
\begin{table}[htbp]
\centering
\vskip -0.1in
\caption{Average $\Delta \Delta G$ (kcal/mol) and average number of changed residues ($\Delta L$). dyMEAN-$n$ denotes the restricted version allowing at most $n$ changed residues, and dyMEAN itself changes $n$ residues with $n$ is sampled from $[1,N]$ at each generation.}
\label{tab:aff_opt_foldx}
\setlength\tabcolsep{3pt}
\scalebox{0.8}{\begin{tabular}{ccc}
\toprule
Method       & $\Delta \Delta G\downarrow$& $\Delta L\downarrow$ \\ \midrule
Diffab ($t=1$)\hphantom{0}  & -0.72             & 1.08       \\
Diffab ($t=2$)\hphantom{0}  & -0.77             & 1.17       \\
Diffab ($t=4$)\hphantom{0}  & -0.42             & 1.10       \\
Diffab ($t=8$)\hphantom{0}  & \hphantom{-}0.07  & 1.42       \\
Diffab ($t=16$)             & \hphantom{-}0.74  & 2.67       \\
Diffab ($t=32$)             & \hphantom{-}1.77  & 6.40       \\ \bottomrule
\end{tabular}}
\hspace{1em}
\scalebox{0.8}{\begin{tabular}{ccc}
\toprule
Method       & $\Delta \Delta G\downarrow$ & $\Delta L\downarrow$ \\ \midrule
MEAN         &   -5.84           &   5.09     \\
dyMEAN-$1$     & -7.95             & 0.94       \\
dyMEAN-$2$     & \textbf{-8.36}    & 1.27       \\
dyMEAN-$4$     & -8.33             & 2.08       \\
dyMEAN-$8$     & -7.89             & 4.37       \\ 
dyMEAN\hphantom{-0} & -8.10             & 2.76       \\ \bottomrule
\end{tabular}}
\vskip -0.15in
\end{table}

\section{Trial on General Proteins}
\label{app:cath}
We have further conducted experimental validation on the CATH dataset. We select the complexes from the CATH dataset and divide each of them into a receptor and a ligand. We identify the interacting residues of the ligands and mask them for generation, with the rest of the ligand as the "framework". Two types of settings are applied for validation. In the "e2e" setting, neither the structure nor the docking position of the framework is provided, resembling the setting in this paper. In the "inpainting" setting, both the structure and the docking position of the framework are provided, thus the model only focuses on filling in the interacting residues and adjusting the framework structure according to the docking position, imitating the setting in previous works (e.g. DiffAb, MEAN). The total number of valid data is 6883. In the e2e setting, we use ESMFold~\citep{lin2023evolutionary} to provide the initial structures in place of the method in \textsection~\ref{sec:init}. Results in Table~\ref{tab:cath} illustrate that dyMEAN achieves promising performance in inpainting the protein complexes. Nevertheless, the performance on the e2e setting reveals greater challenges in designing binding interfaces without prior knowledge on the binding positions. Augmenting data or upgrading models are still urgent needs in this domain.

\begin{table}[htbp]
    \centering
    \caption{Evaluation on general proteins in the CATH dataset.}
    \label{tab:cath}
    \begin{tabular}{lccccc}
        \toprule
        Task        & AAR$\uparrow$ & TMscore$\uparrow$ & lDDT$\uparrow$    & RMSD$\downarrow$  & DockQ$\uparrow$ \\\midrule
        e2e         & 16.02\%       & 0.7678            & 0.6579            & 11.28             & 0.188 \\
        inpainting  & 51.79\%       & 0.9708            & 0.8481            & \hphantom{0}0.66  & 0.916 \\\bottomrule
    \end{tabular}
\end{table}

\section{Samples}
\label{app:sample}
We provide more samples of generated antibodies from epitope-binding CDR-H3 design (\textsection~\ref{sec:cdrh3}) in Figure~\ref{fig:sample}.
\begin{figure}[htbp]
    \centering
    \includegraphics[width=.85\textwidth]{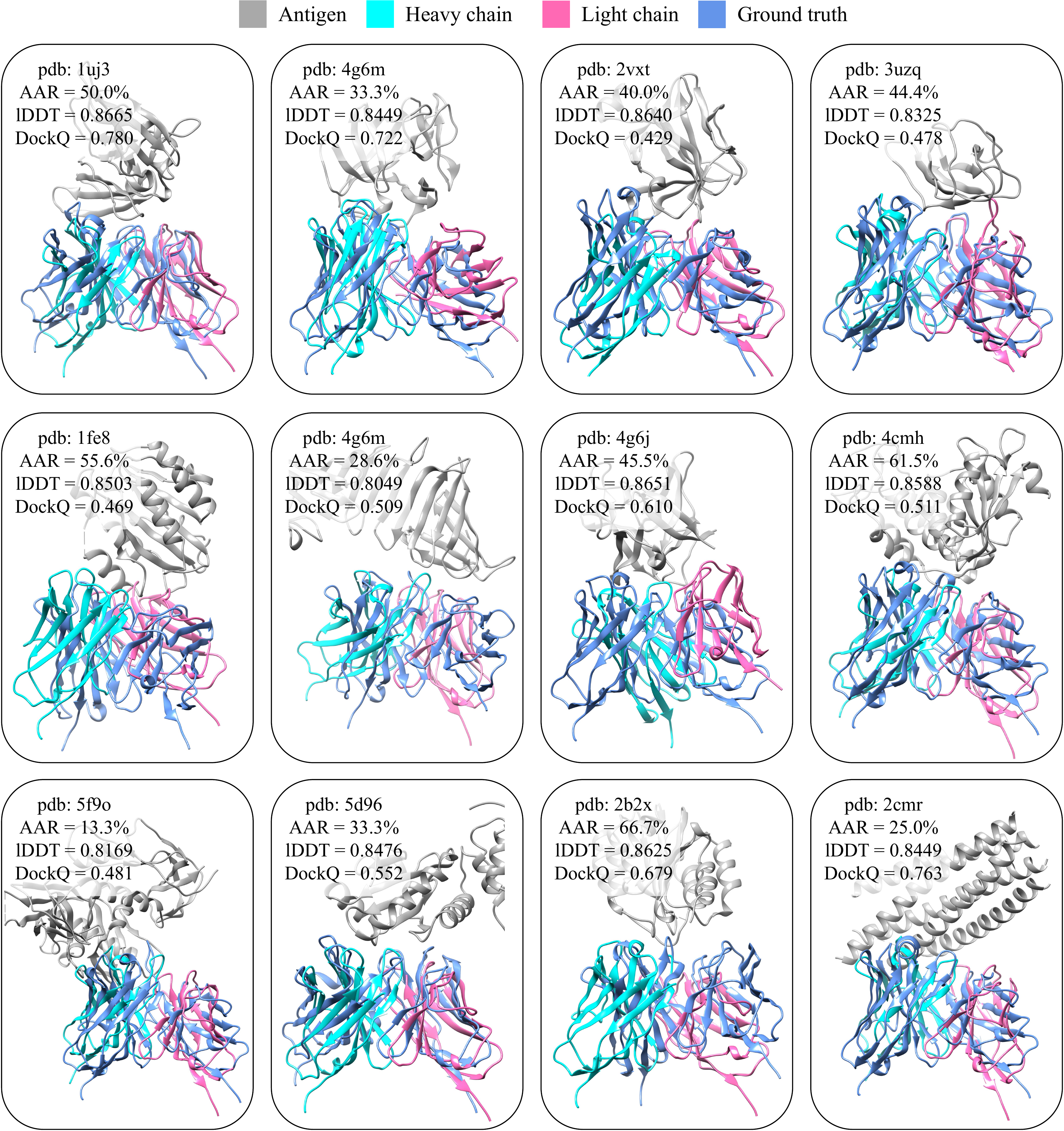}
    \caption{Samples of generated antibodies.}
    \label{fig:sample}
\end{figure}

%

\end{document}